\documentclass[12pt]{article}

\usepackage{amsmath}
\usepackage{amssymb}
\usepackage{amsfonts}
\usepackage{latexsym}
\usepackage{color}
\usepackage{soul}

\usepackage{graphicx}
\usepackage{subfigure}
\usepackage{color,epstopdf}
\catcode `\@=11 \@addtoreset{equation}{section}

\catcode `\@=12

\newcommand{\be}{\begin{equation}}
\newcommand{\en}{\end{equation}}
\newcommand{\bea}{\begin{eqnarray}}
\newcommand{\ena}{\end{eqnarray}}
\newcommand{\beano}{\begin{eqnarray*}}
\newcommand{\enano}{\end{eqnarray*}}
\newcommand{\bee}{\begin{enumerate}}
\newcommand{\ene}{\end{enumerate}}

\newcommand{\mc}{\mathcal}

\newcommand{\D}{{\mc D}}

\newcommand{\Sc}{{\cal S}}
\newcommand{\E}{{\cal E}}
\newcommand{\F}{{\cal F}}
\newcommand{\G}{{\cal G}}

\newcommand{\Lc}{{\cal L}}
\newcommand{\ltwo}{{\Lc^2(\mathbb{R})}}

\newcommand{\kt}{\rangle}
\newcommand{\br}{\langle}

\newcommand{\1}{1 \!\! 1}

\newcommand{\Hil}{\mc H}

\newtheorem{thm}{Theorem}

\newtheorem{prop}[thm]{Proposition}
\newtheorem{defn}[thm]{Definition}

\newenvironment{proof}{\noindent {\bf Proof --}}{\hfill$\square$ \vspace{3mm}\endtrivlist}

\catcode `\@=11 \@addtoreset{equation}{section}
\catcode `\@=12

\begin{document}

\thispagestyle{empty}

\vspace*{-1cm}
\noindent This is the pre-peer reviewed version of the article published at https://doi.org/10.1088/1751-8121/ad02ec \\
%Journal: Z. Angew. Math. Phys.
%(2022) 73:119\\
%Copyright \textcopyright 2022, The Author(s)\\
%Publisher version: https://doi.org/10.1111/sapm.12566\\
%

\begin{center}
{\Large \bf Extended coupled SUSY, pseudo-bosons and weak squeezed states}   \vspace{2cm}\\

{\large F. Bagarello}\\
  Dipartimento di Ingegneria,
Universit\`a di Palermo,\\ I-90128  Palermo, Italy\\
and I.N.F.N., Sezione di Catania\\
e-mail: fabio.bagarello@unipa.it\\

\vspace{2mm}

{\large F. Gargano}\\
Dipartimento di Ingegneria,
Universit\`a di Palermo,\\ I-90128  Palermo, Italy\\
e-mail: francesco.gargano@unipa.it\\

\vspace{2mm}

{\large L. Saluto}\\
Dipartimento di Ingegneria,
Universit\`a di Palermo,\\ I-90128  Palermo, Italy\\
e-mail: lidia.saluto@unipa.it\\

\end{center}

\vspace*{-0cm}

\begin{abstract}
\noindent In this paper we consider squares of pseudo-bosonic ladder operators and we use them to produce explicit examples of eigenstates of certain operators satisfying a deformed  $\mathfrak{su}(1,1)$ Lie algebra. We show how these eigenstates may, or may not, be square-integrable. In both cases, a notion of biorthonormality can be introduced and analyzed. Some examples are discussed in details. We also propose some preliminary results on bi-squeezed states arising from our operators.
\end{abstract}
\vspace{0.5cm}
\textbf{Short title:}\textit{ Coupled SUSY, PBs and weak squeezed states}

\vspace{0.5cm}
\noindent\textbf{Keywords:}
Extended SUSY, Pseudo-bosonic operators, Compatible generalized eigenstates,
Weak squeezed states
%{\bf PACS Numbers}:  .......

\vfill

%\pagenumbering{roman}

\newpage

\section{Introduction}

Finding eigenvalues and eigenvectors of the Hamiltonian operator $H$ of a given physical system $\Sc$ is often the first step towards a full understanding of $\Sc$, and of its dynamical features. However, most of the times, this operation is not so easy. In fact, the eigensystem of a given operator can be exactly deduced only for some operators. Quite often, when $H$ is {\em not so easy}, some approximation methods, or perturbation techniques, need to be adopted.

With the aim of finding more and more {\em solvable Hamiltonians}, many general strategies have been proposed along the years, studied in some details and applied to some specific physical systems. Intertwining operators, factorizable Hamiltonians, ladder operators of different nature, super-symmetric quantum mechanics (Susy qm), are just some of them. We refer to \cite{dong}-\cite{curado2} for many details on these approaches, and for more references.

In this paper, we will focus on the use of ladder operators of a special kind, obeying pseudo-bosonic rules, \cite{bagspringer}, which are essentially an extended version of the canonical commutation relations (CCR) typical of bosonic systems, \cite{bagspringer}. These pseudo-bosonic ladder operators, often simply called $\D$-pseudo bosons\footnote{Here $\D$ is some dense subspace of the Hilbert space where these ladder operators operate.} ($\D$-PBs), will be used in connection with a general strategy proposed in \cite{bag2021JPA}, and which extends similar ideas first proposed in \cite{coupledsusy}, in which the ladder operators obey an $\mathfrak{su}(1,1)$ Lie algebra. The extension is useful in connection with  non self-adjoint Hamiltonians of the kind considered, e.g., in \cite{benbook,bagbookPT}. Our analysis will show how to produce eigenstates of some {\em number-like} operators using suitable quadratic combinations of pseudo-bosonic operators. An interesting aspect of our construction is that some sort of extended biorthogonality can be established,  not between functions of $\ltwo$, but between {\em compatible} functions, i.e. between functions whose product is in $\Lc^1(\mathbb{R})$. This is a feature which has already been observed in different systems, mostly connected to pseudo-bosonic operators, \cite{bag2022JMAA,bagJPCS2021,baginvhosc,bagzamp2023}, where (extended) eigenstates of some specific number-like operator can be found outside $\ltwo$, and for which a discrete set of (again, extended) eigenvalues can be found. This feature will also be found in this paper, for the explicit examples we will discuss later on.

An intriguing possibility that comes from the framework we shall introduce, is the derivation of a weak variant of the (bi-)squeezed states. These states naturally emerge as ground states of two pseudo-bosonic, non-self adjoint, Hamiltonians, constructed using the triplet of operators defining the extended $\mathfrak{su}(1,1)$ Lie algebra considered in this paper. Similarly to the weak formulation of coherent states, as noted in \cite{bag2022JMAA,BG22}, careful considerations are required when defining squeezed states from functions outside the space $\ltwo$. Additionally, within the setting of non-self adjoint Hamiltonians, the loss of regularity in the pseudo-bosonic operators may pose significant issues in defining these states, as evidenced by the case of the Swanson's Hamiltonian, \cite{bisquizi}. Nonetheless, we will demonstrate that by working with compatible functions, we can finally give a suitable definition of the squeezed states in a distributional sense.

 The paper is organized as follows: Section \ref{sect2} is a short review of all the already known results which are relevant for the rest of the paper. In particular, in Section \ref{sect2.2} we give some results and definitions connected to the $\mathfrak{su}(1,1)$ Lie algebra, while in Section \ref{sectPBs} we show how pseudo-bosonic ladder operators give rise to  $\mathfrak{su}(1,1)$. In Section \ref{sectgenclasspbs} we consider a special class of pseudo-bosons,  focusing in particular on three specific examples. We also briefly discuss, in Section \ref{sectsquizi}, some preliminary results on bi-squeezed states, bringing forward the analysis begun in \cite{bisquizi}. Section \ref{sectconcl} contains our conclusions and plans for the future.

\section{Preliminaries}\label{sect2}

In this section we will review some of the results deduced in  \cite{bag2021JPA} and \cite{bag2022JMAA}, in order to make the paper essentially self-contained. In particular, in what follows we will briefly describe our definition of {\em extended coupled Susy} (ECSusy), and its link with a deformed version of the $\mathfrak{su}(1,1)$ Lie algebra. 

\vspace{2mm}

{\bf Remark:--} Before starting our review, it may be useful to stress that the operators involved in the rest of our discussion are unbounded. Because of that, they suffer of serious  domain issues. It is well known that these issues can be usually solved in several ways. We will adopt here an implicit algebraic approach, proposed first in \cite{bag2021JPA}, which assumes that all the operators which are useful for us belong to a certain algebra of unbounded operators, $\Lc^\dagger(\D)$, \cite{aitbook}. Because of that, operators in $\Lc^\dagger(\D)$ can be multiplied or raised to any power, still remaining in $\Lc^\dagger(\D)$. In what follows, this is the only relevant role of this algebra. We refer to \cite{aitbook} for many more details, and to \cite{bag2021JPA} for its explicit construction for our situation. We should also mention that another possible approach to deal with these operators is to assume that it exists a dense subspace of the Hilbert space $\Hil$ where the physical system is defined which is stable under the action of all the operators which, again, are relevant for us. This is indeed the point of view adopted, for instance, in \cite{bagspringer}.   

\subsection{Extending CSusy}\label{sect2.2}

\begin{defn}\label{defecsusy}
	Let $d$, $c$, $r$ and $s$ be four elements of  $\Lc^\dagger(\D)$, and let $\gamma, \delta$ be two real numbers with $\delta>\gamma$. We say that $(d,c,r,s;\delta,\gamma)$ define an {\em extended coupled Susy} (ECSusy), if the following equalities are satisfied:
	\be
	\left\{
	\begin{array}{ll}
		dc=rs+\gamma\1,\\
		cd=sr+\delta\1,\\
	\end{array}
	\right.
	\label{25a}\en
	
\end{defn}
Here $\1$ is the identity operator on $\Hil$, and the formulas above are well defined since $\Lc^\dagger(\D)$ is an algebra. Alternatively, see the last part of our previous remark, equation (\ref{25a})
could be understood as follows: $d(cf)=r(sf)+\gamma f$ and $c(df)=s(rf)+\delta f$, for all $f\in\D$. These equalities are both well defined since, if $f\in\D$, then $cf,df,sf,rf\in\D$ as well, and, therefore, we also have $c(df)\in\D$, and so on.

Let us define the following operators, which are still in $\Lc^\dagger(\D)$:
\be
k_+=\frac{1}{\delta-\gamma}ds, \qquad k_-=\frac{1}{\delta-\gamma}rc, \qquad k_0=\frac{1}{\delta-\gamma}\left(dc-\frac{\gamma}{2}\,\1\right), 
\label{26a}\en
and
\be
l_+=\frac{1}{\delta-\gamma}sd, \qquad l_-=\frac{1}{\delta-\gamma}cr, \qquad l_0=\frac{1}{\delta-\gamma}\left(sr+\frac{\delta}{2}\,\1\right).
\label{27a}\en
Using (\ref{25a}) we find that
\be
[k_0,k_\pm]=\pm k_\pm, \qquad [k_+,k_-]=-2k_0.
\label{28a}\en
Similarly:
\be
[l_0,l_\pm]=\pm l_\pm, \qquad [l_+,l_-]=-2l_0.
\label{29a}\en
Notice that $k_+$ and $l_+$ are not the adjoint of $k_-$ and $l_-$, and $k_0$ and $l_0$ are not self-adjoint, as it happens for the {\em ordinary} $\mathfrak{su}(1,1)$ Lie algebra. This gives us the possibility to introduce two, in general,  different families of operators, $p_\alpha$ and $q_\alpha$, $\alpha=0,\pm$:
\be
p_0=k_0^\dagger, \qquad p_\pm=k_\mp^\dagger; \qquad\qquad q_0=l_0^\dagger, \qquad q_\pm=l_\mp^\dagger.
\label{210a}\en
They satisfy the same commutators in (\ref{28a}) and (\ref{29a}):
\be
[p_0,p_\pm]=\pm p_\pm, \qquad [p_+,p_-]=-2p_0; \qquad [q_0,q_\pm]=\pm q_\pm, \qquad [q_+,q_-]=-2q_0.
\label{211a}\en
Hence we conclude that (\ref{25a}) implies the existence of four (again, in general) different triples of operators obeying the same commutators of an $\mathfrak{su}(1,1)$ Lie algebra, but with different relations under the adjoint operation. 

\subsubsection{The eigenstates of a deformed $\mathfrak{su}(1,1)$ Lie algebra}\label{sectdla}

As in \cite{bag2021JPA}, let us consider here three operators, $x_\pm$ and $x_0$, in $\Lc^\dagger(\D)$, satisfying  $[x_0,x_\pm]=\pm x_\pm$, and $[x_+,x_-]=-2x_0$, but with $x_+^\dagger\neq x_-$ and $x_0^\dagger\neq x_0$. 

First we put, with a slight abuse of notation, 
\be
x^2=x_0^2-\frac{1}{2}(x_+x_-+x_-x_+)=x_0^2+x_0-x_-x_+=x_0^2-x_0-x_+x_-.
\label{212a}\en
We call it {\em an abuse} since $x^2$ is not really the square of an operator $x$ (to be identified). Also, $x^2$ in (\ref{212a}) is not even positive. Nevertheless, we use this notation since it is the one usually adopted in the literature for the {\em ordinary} $\mathfrak{su}(1,1)$ Lie algebra. The operator $x^2$ commutes with each $x_\alpha$: $[x^2,x_\alpha]=0$, for $\alpha=0,\pm$.  Now, since in particular $x^2$ and $x_0$ commute, we can look for common eigenstates of these two operators. Using again the same notation adopted for ordinary $\mathfrak{su}(1,1)$, we assume the following: there exists a non zero vector $\Phi_{j,q_0}\in\D$ satisfying the following eigenvalue equations:
\be
\left\{
\begin{array}{ll}
	x^2\Phi_{j,q_0}=j(j+1)\Phi_{j,q_0},\\
	x_0\Phi_{j,q_0}=q_0\Phi_{j,q_0},\\
\end{array}
\right.
\label{213a}\en
for some $j$ and $q_0$. We should stress that, in principle, there is no reason a priori to assume here that $j$ and $q_0$ are real or positive. This is because, as already observed, $x^2$ is not positive or self-adjoint, and $x_0$ is not self-adjoint. This makes in general much more complicated to describe the set of possible values of $j$ and $q_0$ in (\ref{213a}). However, some useful result can still be found, as we will see. In particular,
\be
\left\{
\begin{array}{ll}
	x^2(x_\pm\Phi_{j,q_0})=j(j+1)(x_\pm\Phi_{j,q_0}),\\
	x_0(x_\pm\Phi_{j,q_0})=(q_0\pm1)(x_\pm\Phi_{j,q_0}),\\
\end{array}
\right.
\label{214a}\en
at least if $\Phi_{j,q_0}\notin \ker(x_\pm)$. This means that $x_\pm$ are ladder operators and, in particular, that $x_+$ is a raising while $x_-$ is a lowering operator. Using the same standard arguments for $\mathfrak{su}(1,1)$, we can also deduce that
\be
\left\{
\begin{array}{ll}
	x_+\Phi_{j,q_0}=(q_0-j)\Phi_{j,q_0+1},\\
	x_-\Phi_{j,q_0}=(q_0+j)\Phi_{j,q_0-1}.\\
\end{array}
\right.
\label{215a}\en
These equations are in agreement with the fact that, as it is easy to check,
$$
[x_0,x_-x_+]=[x_0,x_+x_-]=0.
$$
In fact, from (\ref{215a}) we see that $x_0$ and $x_-x_+$ have the same eigenvectors. The same is true for  $x_0$ and $x_+x_-$.

As we have already observed in \cite{bag2021JPA}, we have several possibilities: 

\vspace{2mm}

{\bf Case 1:--} for some $m\in\mathbb{N}_0=\mathbb{N}\cup\{0\}$ we have $x_-^{m-1}\Phi_{j,q_0}\neq0$ and $x_-^{m}\Phi_{j,q_0}=0$. In this case the set of eigenvalues of $x_0$, $\sigma(x_0)$, is bounded below: $\sigma(x_0)=\{q_0-m+1,q_0-m+2,q_0-m+3,\ldots\}$.

\vspace{2mm}

{\bf Case 2:--} for some $k\in\mathbb{N}_0=\mathbb{N}\cup\{0\}$ we have $x_+^{k-1}\Phi_{j,q_0}\neq0$ and $x_+^{k}\Phi_{j,q_0}=0$. In this case $\sigma(x_0)$, is bounded above: $\sigma(x_0)=\{\ldots,q_0+k-3,q_0+k-2,q_0+k-1\}$.

\vspace{2mm}

{\bf Case 3:--} both conditions above are true. In this case, of course, $\sigma(x_0)$, is bounded above and below: $\sigma(x_0)=\{q_0-m+1,q_0-m+2,\ldots,q_0+k-2,q_0+k-1\}$.

\vspace{2mm}

{\bf Case 4:--} neither Case 1, nor Case 2, hold. Then $\sigma(x_0)$ has no bound below and above.

\subsubsection{Back to ECSusy}\label{sectBECSusy}

We can use now these general results in the analysis of the operators introduced in Section \ref{sect2.2}. However, this will not be the only ingredient of the procedure we are going to propose. In fact, as we will see, the natural biorthonormality connected to the appearance of non self-adjoint number-like operators will play a relevant role. We first consider the operators $k_\alpha$, $\alpha=0,\pm$. As in (\ref{213a}), we assume a non zero vector $\varphi_{j,q}\in\D$ exists, $j,q\in\mathbb{C}$, such that
\be
k^2\varphi_{j,q}=j(j+1)\varphi_{j,q}, \qquad k_0\varphi_{j,q}=q\varphi_{j,q}.
\label{216}\en
Here, as in (\ref{212a}), $k^2=k_0^2+k_0-k_-k_+$, for instance. The operators $k_\pm$ act on $\varphi_{j,q}$ as ladder operators:
\be
k_+\varphi_{j,q}=(q-j)\varphi_{j,q+1}, \qquad k_-\varphi_{j,q}=(q+j)\varphi_{j,q-1}, 
\label{217}\en for all $\varphi_{j,q}\notin\ker(k_\pm)$. Let us now call $I_j$ the set of all the $q's$ for which $\varphi_{j,q}$ is not annihilated by at least one between $k_+$ and $k_-$: if $q\in I_j$, then $\varphi_{j,q}\notin\ker(k_+)$ or $\varphi_{j,q}\notin\ker(k_-)$, or both, and let $\F_\varphi(j):=\{\varphi_{j,q}, \,\forall q\in I_j\}$. Let then introduce $\E_j=l.s.\{\varphi_{j,q}, \, q\in I_j\}$, the linear span of the vectors in $\F_\varphi(j)$, and $\Hil_j$ the closure of $\E_j$, with respect to the norm of $\Hil$. Of course, $\Hil_j\subseteq\Hil$, for each fixed $j$. By construction, $\F_\varphi(j)$ is a basis for $\Hil_j$. Let $\F_\psi(j):=\{\psi_{j,q}, \,\forall q\in I_j\}$ be its unique biorthogonal basis, \cite{chri}. Then 
\be
\br\varphi_{j,q},\psi_{j,r}\kt=\delta_{q,r}, 
\label{218}\en
for all $q,r\in I_j$, and  $l.s.\{\psi_{j,q}, \, q\in I_j\}$ is dense in $\Hil_j$. Using (\ref{210a}), it is possible to check the following eigenvalue and ladder equalities:
\be
\left\{
\begin{array}{ll}
	p_0\psi_{j,q}=\overline{q}\psi_{j,q},\\
	p_+\psi_{j,q}=\overline{(q+1+j)}\psi_{j,q+1},\\
	p_-\psi_{j,q}=\overline{(q-1-j)}\psi_{j,q+1},\\
\end{array}
\right.
\label{219}\en
at least if $\psi_{j,q}\notin\ker(p_\pm)$. We refer to \cite{bag2021JPA} for the reasons why these formulas look slightly different from those in 
 (\ref{215a}). Here we just observe that the vectors $\varphi_{j,q}$ and the $\psi_{j,q}$ are in fact introduced in the game in a different way: while $\varphi_{j,q}$ are those vectors satisfying (\ref{216}), $\{\psi_{j,q}\}$ is the unique set of vectors which is biorthonormal to $\{\varphi_{j,q}\}$. This is, in fact, the reason of the difference we observe in the ladder equations.
 
 \vspace{2mm}
 
 {\bf Remark:--} It is useful to anticipate that, in the examples discussed in Section \ref{sectgenclasspbs}, the role of the Hilbert space $\ltwo$ will often be marginal. In that case, more than the Hilbert space framework, we will use a (simpler) vector space settings. But, as already stressed, the concept of biorthogonality could still be used.

 \vspace{2mm}
 
In \cite{bag2021JPA} we have also deduced many intertwining equations. For instance, we have
\be
\left\{
\begin{array}{ll}
	sk_+=l_+s,\qquad\qquad\qquad k_+d=dl_+\qquad\qquad\qquad l_0s=s\left(k_0+\frac{1}{2}\1\right)\\
	ck_-=l_-c,\qquad\qquad\qquad k_-r=rl_-\qquad\qquad\qquad l_0c=c\left(k_0-\frac{1}{2}\1\right)\\
	r^\dagger p_+=q_+r^\dagger,\qquad\qquad\quad p_+c^\dagger=c^\dagger q_+\qquad\qquad\quad rl_0=\left(k_0+\frac{1}{2}\1\right)r\\
	d^\dagger p_-=q_-d^\dagger,\qquad\qquad\quad p_-s^\dagger=s^\dagger q_-\qquad\qquad\quad dl_0=\left(k_0-\frac{1}{2}\1\right)d,
\end{array}
\right.
\label{220}\en
Similar intertwining equations could be found for their adjoint, \cite{bag2021JPA}.

\vspace{2mm}

Among the other results, an interesting consequence of (\ref{220}) is that the eigenvalues of $l_0$ differ from those of $k_0$ by half integers, as those of $q_0$ from those of $p_0$. Indeed we have, considering a vector $\varphi_{j,q}$ with  $s\varphi_{j,q}\neq0$ and $c\varphi_{j,q}\neq0$,
$$
l_0\left(s\varphi_{j,q}\right)=s\left(k_0+\frac{1}{2}\1\right)\varphi_{j,q}=\left(q+\frac{1}{2}\right)\left(s\varphi_{j,q}\right).
$$
as well as
$$
l_0\left(c\varphi_{j,q}\right)=\left(q-\frac{1}{2}\right)\left(c\varphi_{j,q}\right).
$$
Many other details of this construction can be found in \cite{bag2021JPA}.

\subsection{A detailed example: $\D$-PBs}\label{sectPBs}

In \cite{bag2021JPA} we have shown how $\D$-PBs can be used to generate examples of the above framework. Some of these results will be briefly outlined in this section. In Section \ref{sectgenclasspbs} we will extend our results outside Hilbert spaces, and we will propose three different concrete examples.

First we recall what $\D$-PBs are: let $\Hil$ be a given Hilbert space with scalar product $\left<.,.\right>$ and related norm $\|.\|$. Let $a$ and $b$ be two operators
on $\Hil$, with domains $D(a)\subset \Hil$ and $D(b)\subset \Hil$ respectively, $a^\dagger$ and $b^\dagger$ their adjoint, and let $\D$ be a dense subspace of $\Hil$
such that $a^\sharp\D\subseteq\D$ and $b^\sharp\D\subseteq\D$. Here with $x^\sharp$ we indicate $x$ or $x^\dagger$. Of course, $\D\subseteq D(a^\sharp)$
and $\D\subseteq D(b^\sharp)$.

\begin{defn}\label{def21}
	The operators $(a,b)$ are $\D$-pseudo bosonic  if, for all $f\in\D$, we have
	\be
	a\,b\,f-b\,a\,f=f.
	\label{A1}\en
\end{defn}

\vspace{2mm}

We further require, \cite{bagspringer}, the existence of two non zero vectors,  $\varphi_{ 0}, \Psi_0\in\D$ such that $a\,\varphi_{ 0}=b^\dagger\Psi_0=0$.

\vspace{2mm}

The invariance of $\D$ under the action of $b$ and $a^\dagger$ implies 
that the vectors \be \varphi_n:=\frac{1}{\sqrt{n!}}\,b^n\varphi_0,\qquad \Psi_n:=\frac{1}{\sqrt{n!}}\,{a^\dagger}^n\Psi_0, \label{31}\en
$n\geq0$, are well defined and they all belong to $\D$ and, as a consequence, to the domain of $a^\sharp$, $b^\sharp$ and $N^\sharp$, where $N=ba$. Let us put $\F_\Psi=\{\Psi_{ n}, \,n\geq0\}$ and
$\F_\varphi=\{\varphi_{ n}, \,n\geq0\}$.
It is  simple to deduce the following lowering and raising relations:
\be
\left\{
\begin{array}{ll}
	b\,\varphi_n=\sqrt{n+1}\varphi_{n+1}, \qquad\qquad\quad\,\, n\geq 0,\\
	a\,\varphi_0=0,\quad a\varphi_n=\sqrt{n}\,\varphi_{n-1}, \qquad\,\, n\geq 1,\\
	a^\dagger\Psi_n=\sqrt{n+1}\Psi_{n+1}, \qquad\qquad\quad\, n\geq 0,\\
	b^\dagger\Psi_0=0,\quad b^\dagger\Psi_n=\sqrt{n}\,\Psi_{n-1}, \qquad n\geq 1,\\
\end{array}
\right.
\label{32a}\en as well as the eigenvalue equations $N\varphi_n=n\varphi_n$ and  $N^\dagger\Psi_n=n\Psi_n$, $n\geq0$. Then,  if we choose the normalization of $\varphi_0$ and $\Psi_0$ in such a way $\left<\varphi_0,\Psi_0\right>=1$, we deduce that
\be \left<\varphi_n,\Psi_m\right>=\delta_{n,m}, \label{33a}\en
for all $n, m\geq0$. Hence $\F_\Psi$ and $\F_\varphi$ are biorthonormal. In \cite{baginbagbook} it is shown that, in several quantum models, $\F_\Psi$ and $\F_\varphi$ are complete in $\Hil$, but they are not bases. However, they still produce useful resolutions of the identity since they are always, at least in all the systems considered so far, $\G$-quasi bases, where $\G$ is some subspace dense in $\Hil$. This means that
for all $f$ and $g$ in $\G$,
\be
\left<f,g\right>=\sum_{n\geq0}\left<f,\varphi_n\right>\left<\Psi_n,g\right>=\sum_{n\geq0}\left<f,\Psi_n\right>\left<\varphi_n,g\right>.
\label{34a}
\en

We refer to \cite{baginbagbook} for many more results and examples on $\D$-quasi bosons. Here, what is relevant for us, are the ladder properties described by (\ref{32a}), and the fact that they produce a concrete, and highly non trivial, example of ECSusy.

In fact, let us fix the operators and the numbers $\gamma$ and $\delta$ in Definition \ref{defecsusy} as follows: $c=r=a$, $d=s=b$, $\delta=-\gamma=1$, where $a$ and $b$ satisfy Definition \ref{def21}. Hence the operators in (\ref{26a}), (\ref{27a}) and (\ref{210a}) become
\be
k_+=l_+=\frac{1}{2}b^2, \qquad k_-=l_-=\frac{1}{2}a^2, \qquad k_0=l_0=\frac{1}{2}\left(N+\frac{1}{2}\1\right),
\label{35a}\en
where $N=ba$, and
\be
p_+=q_+=\frac{1}{2}{a^\dagger}^2, \qquad p_-=q_-=\frac{1}{2}{b^\dagger}^2, \qquad p_0=q_0=\frac{1}{2}\left(N^\dagger+\frac{1}{2}\1\right).
\label{36a}\en
It is clear that the four original families collapse into two. Formula (\ref{212a}) produce further the following result:
\be
k^2=p^2=-\frac{3}{16}\1,
\label{37a}\en
which, of course, commute with all the other operators, as expected. We notice that this formula clarifies what already observed  after formula (\ref{212a}): despite of their "names", $k^2$ and $p^2$ are not positive operators. Formula (\ref{213a}) is based on the assumption that a non zero eigenstate of $x^2$ and $x_0$ exists. In our situation,  such a vector can be easily found: in fact, if we consider the vacuum $\varphi_0$ introduced before, see Assumption $\D$-pb 1., we have $$k^2\varphi_0=-\frac{3}{16}\varphi_0, \qquad k_0\varphi_0=\frac{1}{4}\varphi_0.$$
Hence, comparing these with (\ref{213a}), we have $q_0=\frac{1}{4}$ and $j(j+1)=-\frac{3}{16}$, that is $j=-\frac{1}{4}$ or $j=-\frac{3}{4}$. Because of formula (\ref{217}), and observing that $k_-\varphi_0=0$, we choose 
$j=-\frac{1}{4}$ and we define
\be
\varphi_{-\frac{1}{4},\frac{1}{4}}:=\varphi_0.
\label{38a}\en
Hence we are in Case 1 of Section \ref{sectdla}, with $m=1$. In fact, since the spectrum of $N$ is the set $\mathbb{N}_0=\mathbb{N}\cup\{0\}$, $\sigma(k_0)$ is bounded below.

If we act $m$ times with $k_+$ on $\varphi_{-\frac{1}{4},\frac{1}{4}}$, $m=1,2,3,\ldots$, formula (\ref{217}) produces
\be
\varphi_{-\frac{1}{4},m+\frac{1}{4}}=\frac{\sqrt{(2m)!}}{(2m-1)!!}\varphi_{2m},
\label{39a}\en
where $\varphi_{2m}$ are those in (\ref{31}) and, with standard notation, $(2m-1)!!=1\cdot3\cdots(2m-3)\cdot(2m-1)$, with $0!!=(-1)!!=1$. Using (\ref{216}) and (\ref{217}), or with a direct check, we find
\be
k_0\varphi_{-\frac{1}{4},m+\frac{1}{4}}=\left(m+\frac{1}{4}\right)\varphi_{-\frac{1}{4},m+\frac{1}{4}},
\label{310}\en
and
\be
k_+\varphi_{-\frac{1}{4},m+\frac{1}{4}}=\left(m+\frac{1}{2}\right)\varphi_{-\frac{1}{4},m+\frac{5}{4}}, \qquad k_-\varphi_{-\frac{1}{4},m+\frac{1}{4}}=m\,\varphi_{-\frac{1}{4},m-\frac{3}{4}}.
\label{311}\en
In particular, this last equality is true only if $m\geq1$. If $m=0$ we have $k_-\varphi_{-\frac{1}{4},\frac{1}{4}}=k_-\varphi_0=0$, as already noticed.

According to Section \ref{sectBECSusy}, we can now define the set of linearly independent vectors $\F_\varphi^{(e)}\left(\frac{1}{4}\right)=\{\varphi_{-\frac{1}{4},m+\frac{1}{4}}, \, m=0,1,2,3,\ldots\}$, and the Hilbert space $\Hil_{-\frac{1}{4}}^{(e)}$, constructed by taking the closure of the linear span of these vectors. Here the suffix {\em e} stands for {\em even}, since only the vectors $\varphi_{2m}$ belong to $\F_\varphi^{(e)}\left(\frac{1}{4}\right)$. It is clear that $\Hil_{-\frac{1}{4}}^{(e)}\subset\Hil$, since all the vectors with odd index, $\varphi_{2m+1}$, which belong to $\Hil$, do not belong to $\Hil_{-\frac{1}{4}}^{(e)}$. Hence, the set 
$\F_\varphi^{(e)}\left(\frac{1}{4}\right)$ cannot be complete in $\Hil$, and, as a consequence, cannot be a basis for $\Hil$. Nevertheless, by construction, $\Hil_{-\frac{1}{4}}^{(e)}$ is an Hilbert space as well, and $\F_\varphi^{(e)}\left(\frac{1}{4}\right)$ is a basis for it. Then, see \cite{chri}, an unique biorthonormal basis $\F_\psi^{(e)}\left(\frac{1}{4}\right)=\{\psi_{-\frac{1}{4},m+\frac{1}{4}}, \, m=0,1,2,3,\ldots\}$ exists, such that
\be
\br\varphi_{-\frac{1}{4},m+\frac{1}{4}},\psi_{-\frac{1}{4},l+\frac{1}{4}}\kt=\delta_{m,l},
\label{312}\en
where the scalar product is the one in $\Hil$, and, for each $f\in \Hil_{-\frac{1}{4}}^{(e)}$,
\be
f=\sum_{m=0}^{\infty}\br\varphi_{-\frac{1}{4},m+\frac{1}{4}},f\kt\,\psi_{-\frac{1}{4},m+\frac{1}{4}}=\sum_{m=0}^{\infty}\br\psi_{-\frac{1}{4},m+\frac{1}{4}},f\kt\,\varphi_{-\frac{1}{4},m+\frac{1}{4}}.
\label{313}\en
From (\ref{39a}) and (\ref{33a}) it is clear that the vectors of this biorthonormal basis are the following:
\be
\psi_{-\frac{1}{4},m+\frac{1}{4}}=\frac{(2m-1)!!}{\sqrt{(2m)!}}\psi_{2m}.
\label{314}\en
Formulas (\ref{219}) can now be explicitly checked, and we get 
\be
p^2\psi_{-\frac{1}{4},m+\frac{1}{4}}=-\frac{3}{16}\psi_{-\frac{1}{4},m+\frac{1}{4}}, \qquad p_0\psi_{-\frac{1}{4},m+\frac{1}{4}}=\left(m+\frac{1}{4}\right)\psi_{-\frac{1}{4},m+\frac{1}{4}}, 
\label{315}\en
together with
\be
p_+\psi_{-\frac{1}{4},m+\frac{1}{4}}=\left(m+1\right)\psi_{-\frac{1}{4},m+\frac{5}{4}}, \qquad p_-\psi_{-\frac{1}{4},m+\frac{1}{4}}=\left(m-\frac{1}{2}\right)\psi_{-\frac{1}{4},m-\frac{3}{4}}.
\label{316}\en

Once more, we stress that the difference between these ladder equations and those in (\ref{311}) arises because, while the $\varphi_{-\frac{1}{4},m+\frac{1}{4}}$'s are introduced using directly the deformed $\mathfrak{su}(1,1)$ algebra, the $\psi_{-\frac{1}{4},m+\frac{1}{4}}$'s are just the unique basis which is biorthonormal to $\F_\varphi^{(e)}\left(\frac{1}{4}\right)$. However, as (\ref{315}) and (\ref{316}) show, these vectors are still eigenstates of $p^2$ and $p_0$, and obey interesting ladder equations with respect to $p_\pm$, which are slightly different from those in (\ref{215a}).

Let us now consider the intertwining relations in (\ref{220}). In the case of $\D$-PBs, these correspond to the following two equalities:
\be
k_0b=b\left(k_0+\frac{1}{2}\1\right), \qquad k_0a=a\left(k_0-\frac{1}{2}\1\right).
\label{317}\en
The consequence of this kind of equalities is well known: if $\rho$ is an eigenstate of $k_0$ with eigenvalue $E$, $k_0\rho=E\rho$, and if $a\rho$ and $b\rho$ are both non zero, then 
$$
k_0(a\rho)=\left(E-\frac{1}{2}\right)(a\rho), \qquad k_0(b\rho)=\left(E+\frac{1}{2}\right)(b\rho),
$$
which means that $a\rho$ and $b\rho$ are both eigenstates of $k_0$, but with two shifted (and different) eigenvalues, $E\pm\frac{1}{2}$. Now, since (\ref{310}) shows that the eigenvalues related to different vectors $\varphi_{-\frac{1}{4},m+\frac{1}{4}}$ and $\varphi_{-\frac{1}{4},l+\frac{1}{4}}$ differ for integer quantities, we conclude that neither $a\varphi_{-\frac{1}{4},m+\frac{1}{4}}$, nor $b\varphi_{-\frac{1}{4},m+\frac{1}{4}}$, can still be of the same form $\varphi_{-\frac{1}{4},l+\frac{1}{4}}$, for any $l\in\mathbb{N}_0$. And, in fact, this can be explicitly checked, since
\be
a\varphi_{-\frac{1}{4},m+\frac{1}{4}}=\sqrt{2m}\frac{\sqrt{(2m)!}}{(2m-1)!!}\,\varphi_{2m-1}, \qquad b\varphi_{-\frac{1}{4},m+\frac{1}{4}}=\frac{\sqrt{(2m+1)!}}{(2m-1)!!}\,\varphi_{2m+1}, 
\label{318}\en
with the agreement that $\varphi_{-1}=0$. Let us now define
\be
\varphi_{-\frac{1}{4},m+\frac{3}{4}}:=b\varphi_{-\frac{1}{4},m+\frac{1}{4}}=\frac{\sqrt{(2m+1)!}}{(2m-1)!!}\,\varphi_{2m+1},
\label{319}\en
for all $m\geq0$. The reason for calling this vector in this way is because $\varphi_{-\frac{1}{4},m+\frac{3}{4}}$ is an eigenstate of $k_0$ with eigenvalue $m+\frac{3}{4}$, as expected because of our previous analysis on $b\rho$:
\be
k_0\varphi_{-\frac{1}{4},m+\frac{3}{4}}=\left(m+\frac{3}{4}\right)\varphi_{-\frac{1}{4},m+\frac{3}{4}},
\label{320}\en
We further deduce the following raising and lowering relations:
\be
k_+\varphi_{-\frac{1}{4},m+\frac{3}{4}}=\left(m+\frac{1}{2}\right)\varphi_{-\frac{1}{4},m+\frac{7}{4}}, \qquad k_-\varphi_{-\frac{1}{4},m+\frac{3}{4}}=m\,\frac{2m+1}{2m-1}\,\varphi_{-\frac{1}{4},m-\frac{1}{4}},
\label{321}\en
with the agreement that $\varphi_{-\frac{1}{4},-\frac{1}{4}}=0$.

\vspace{3mm}

It is clear that, in the same way in which $a$ and $b$ map $\varphi_{-\frac{1}{4},m+\frac{1}{4}}$ into some $\varphi_{-\frac{1}{4},l+\frac{3}{4}}$, they also map these last vectors into the previous ones. More explicitly, we have
\be
a\varphi_{-\frac{1}{4},m+\frac{3}{4}}=\sqrt{2m+1}\,\frac{\sqrt{(2m+1)!}}{(2m-1)!!}\,\varphi_{2m}, \qquad b\varphi_{-\frac{1}{4},m+\frac{3}{4}}=\frac{\sqrt{(2m+2)!}}{(2m-1)!!}\,\varphi_{2m+2}, 
\label{322}\en
for all $m\geq0$. Notice that the vectors in the RHS of these equalities are proportional to $\varphi_{-\frac{1}{4},m+\frac{1}{4}}$ and to $\varphi_{-\frac{1}{4},m+\frac{5}{4}}$, see (\ref{29a}).

In analogy with what we have done before, we introduce now the set $\F_\varphi^{(o)}\left(\frac{1}{4}\right)=\{\varphi_{-\frac{1}{4},m+\frac{3}{4}}, \, m=0,1,2,3,\ldots\}$, where {\em o} stands for {\em odd}, and the Hilbert space $\Hil_{-\frac{1}{4}}^{(o)}$, constructed by taking the closure of the linear span of its vectors. It is clear that $\Hil_{-\frac{1}{4}}^{(e)}\cap\Hil_{-\frac{1}{4}}^{(o)}=\emptyset$, and that, together, $\F_\varphi\left(\frac{1}{4}\right):=\F_\varphi^{(e)}\left(\frac{1}{4}\right)\cup\F_\varphi^{(o)}\left(\frac{1}{4}\right)$ is complete in $\Hil$, at least if the set $\F_\varphi$ is complete, which is always the case in all the concrete examples in the literature, in our knowledge. In particular, if the $\D$-PBs are {\em regular}, see \cite{baginbagbook}, $\F_\varphi$ and $\F_\psi$ are biorthonormal Riesz bases. Hence $\F_\varphi\left(\frac{1}{4}\right)$ is a Riesz basis as well.

Now, since $\F_\varphi^{(o)}\left(\frac{1}{4}\right)$ is a basis for $\Hil_{-\frac{1}{4}}^{(o)}$, we can introduce an unique biorthonormal  basis $\F_\psi^{(o)}\left(\frac{1}{4}\right)=\{\psi_{-\frac{1}{4},m+\frac{3}{4}}, \, m=0,1,2,3,\ldots\}$, whose vectors can be easily identified using (\ref{319}) and (\ref{33a}). We have
\be
\psi_{-\frac{1}{4},m+\frac{3}{4}}=\frac{(2m-1)!!}{\sqrt{(2m+1)!}}\,\psi_{2m+1}=\frac{1}{2m+1}a^\dagger \psi_{-\frac{1}{4},m+\frac{1}{4}}.
\label{323}\en
It may be interesting to notice the difference in the normalization between $\psi_{-\frac{1}{4},m+\frac{3}{4}}$ and $\varphi_{-\frac{1}{4},m+\frac{3}{4}}$, in terms of their $m+\frac{1}{4}$ counterparts, see (\ref{319}) and (\ref{323}). This difference arises because we want to maintain biorthonormality of the vectors. In fact, with the choice in (\ref{323}) we get
\be
\br\varphi_{-\frac{1}{4},m+\frac{3}{4}},\psi_{-\frac{1}{4},l+\frac{3}{4}}\kt=\delta_{m,l},
\label{324}\en
where the scalar product is the one in $\Hil$, and, for each $f\in \Hil_{-\frac{1}{4}}^{(o)}$,
\be
f=\sum_{m=0}^{\infty}\br\varphi_{-\frac{1}{4},m+\frac{3}{4}},f\kt\,\psi_{-\frac{1}{4},m+\frac{3}{4}}=\sum_{m=0}^{\infty}\br\psi_{-\frac{1}{4},m+\frac{3}{4}},f\kt\,\varphi_{-\frac{1}{4},m+\frac{3}{4}}.
\label{325}\en
{Repeating then what we have done for $\Hil^{(e)}$, we can consider the set $\F_\psi^{(o)}\left(\frac{1}{4}\right)=\{\psi_{-\frac{1}{4},m+\frac{3}{4}}, \, m=0,1,2,3,\ldots\}$, and observe that
$\F_\psi\left(\frac{1}{4}\right):=\F_\psi^{(e)}\left(\frac{1}{4}\right)\cup\F_\psi^{(o)}\left(\frac{1}{4}\right)$ is complete in $\Hil$, or it is even a Riesz basis for $\Hil$, depending on the nature of the $\D$-PBs we are considering.}  
%More in detail, if we now introduce the families $\F_\Phi=\{\Phi_k, \, k\geq0\}$ and $\F_\xi=\{\xi_k, \, k\geq0\}$, where
%$$
%\Phi_k=\left\{
%\begin{array}{ll}
%	\varphi_{-\frac{1}{4},j+\frac{1}{4}}, \qquad\quad\,\, \mbox{if } k=2j,\\
%	\varphi_{-\frac{1}{4},j+\frac{3}{4}}, \qquad\quad\,\, \mbox{if } k=2j+1,\\
%\end{array}
%\right.\quad\mbox{ and }\quad \xi_k=\left\{
%\begin{array}{ll}
%	\psi_{-\frac{1}{4},j+\frac{1}{4}}, \qquad\quad\,\, \mbox{if } k=2j,\\
%	\psi_{-\frac{1}{4},j+\frac{3}{4}}, \qquad\quad\,\, \mbox{if } k=2j+1,\\
%\end{array}
%\right.
%$$
%$k\geq0$, we can check that $\br\Phi_k,\xi_l\kt=\delta_{k,l}$, and that, $\forall f,g\in\D$,
%$$
%\sum_{k=0}^{\infty}\br f,\Phi_k\kt\br\xi_k,g\kt= \sum_{k=0}^{\infty}\br f,\varphi_k\kt\br\psi_k,g\kt, \qquad \sum_{k=0}^{\infty}\br f,\xi_k\kt\br\Phi_k,g\kt= \sum_{k=0}^{\infty}\br f,\psi_k\kt\br\varphi_k,g\kt.
%$$
%These equalities imply that $\F_\Phi$ and $\F_\xi$ are biorthonormal, and, \cite{baginbagbook}, that they are $\D$-quasi bases
%if and only if $\F_\varphi$ and $\F_\psi$ are $\D$-quasi bases. We will see later the explicit definition of $\D$ in some particular situation.

\section{A general class of pseudo-bosonic operators}\label{sectgenclasspbs}

The main aim of this paper is to consider what happens, and what must be modified of the general settings described before, if $a$ and $b$ are first order differential operators of the form
\be
a=\alpha_a(x)\,\frac{d}{dx}+\beta_a(x), \qquad b=-\frac{d}{dx}\,\alpha_b(x)+\beta_b(x), 
\label{41}\en
for some suitable, sufficiently regular, functions $\alpha_j(x)$ and $\beta_j(x)$, $j=a,b$

First of all, for those $f(x)$ for which $[a,b]f(x)$ does make sense,  $[a,b]f(x)=f(x)$ if we have, \cite{bag2022JMAA},
\be
\left\{
\begin{array}{ll}
	\alpha_a(x)\alpha_b'(x)=\alpha_a'(x)\alpha_b(x), \\
	\alpha_a(x)\beta_b'(x)+	\alpha_b(x)\beta_a'(x)=1+\alpha_a(x)\alpha_b''(x).\\
\end{array}
\right.
\label{42}\en
In what follows, we will rewrite what has been originally deduced in \cite{bag2022JMAA}, implementing the following simple remark: if we restrict to functions $\alpha_j(x)$ which are never zero, then the first condition in (\ref{42}) can be rewritten as $\frac{d}{dx}\frac{\alpha_a(x)}{\alpha_b(x)}=0$, which implies that $\alpha_a(x)$ and $\alpha_b(x)$ must be proportional. For this reason, from now on and to simplify a little bit the notation, we will write
\be
\alpha_b(x)=\alpha(x), \qquad \alpha_a(x)=k\alpha_b(x)=k\alpha(x),
\label{43b}\en
where $k$ is the proportionality constant which will always be assumed to be real and positive, from now on ($k>0$),  to fix ideas. In this way we can rewrite the second equation in (\ref{42}) as follows:
\be(\beta_a(x)+k\beta_b(x))'=\frac{1}{\alpha(x)}+k\alpha''(x).
\label{44b}\en
In the rest of this paper we will identify $\beta_a(x)$ and $\beta_b(x)$ as follows:
\be
\beta_a(x)=\int\frac{dx}{\alpha(x)}, \qquad \beta_b(x)=\alpha'(x).
\label{45b}\en
Of course, many other possible choices exist. The easiest alternative is when the role of $\beta_a(x)$ and $\beta_b(x)$ are exchanged. But we could also consider $\beta_a(x)=\int\frac{dx}{\alpha(x)}+\Phi(x)$ and $ \beta_b(x)=\alpha'(x)-\Phi(x)$, for any possible choices of (sufficiently regular) $\Phi(x)$. It is clear that the freedom in choosing $\Phi(x)$ makes this settings rather rich, and open the way to many possible interesting situations.  

\vspace{2mm}

{\bf Remark:--} Notice that, fixing $\Phi(x)=0$, the operators in (\ref{41}) can be rewritten as
\be
a=k\,\alpha(x)\,\frac{d}{dx}+\beta_a(x), \qquad b=-\alpha(x)\,\frac{d}{dx}.
\label{45d}\en

\vspace{2mm}

In what follows, we will restrict to (\ref{45b}), postponing our analysis on different choices of $\Phi(x)$ to a future paper. However, we will  keep the more general expressions in (\ref{41}) for $a$ and $b$, rather than (\ref{45d}), since this will be useful when extending our results to $\Phi(x)\neq0$.

\vspace{2mm}

As in \cite{bag2022JMAA} we introduce now the (formal) adjoints of $a$ and $b$ as \be
a^\dagger=-k\frac{d}{dx}\,\overline{\alpha(x)}+\overline{\beta_a(x)}, \qquad b^\dagger=\overline{\alpha(x)}\,\frac{d}{dx}+\overline{\beta_b(x)}.
\label{46b}\en

The vacua of $a$ and $b^\dagger$ are the solutions of $a\varphi_0(x)=0$ and $b^\dagger\psi_0(x)=0$, which are easily found:
\be
\varphi_0(x)=N_\varphi \exp\left\{-\frac{1}{k}\,\int\frac{\beta_a(x)}{\alpha(x)}\,dx\right\}, \qquad \psi_0(x)=N_\psi \exp\left\{-\int\frac{\overline{\beta_b(x)}}{\overline{\alpha(x)}}\,dx\right\},
\label{47b}\en
and are well defined under our assumptions on $\alpha(x)$ and $\beta_j(x)$. Of course, these formulas simplify using (\ref{45b}):

\be
\varphi_0(x)=N_\varphi \exp\left\{-\frac{1}{2\,k}\,(\beta_a(x))^2\right\}  \qquad \psi_0(x)=N_\psi \frac{1}{\overline{\alpha(x)}},
\label{47b-new}
\en

 Here $N_\varphi$ and $N_\psi$ are normalization constants which will be fixed later. If we now introduce $\varphi_n(x)$ and $\psi_n(x)$ as usual,
\be
\varphi_n(x)=\frac{1}{\sqrt{n!}}\,b^n\varphi_0(x),\qquad \psi_n(x)=\frac{1}{\sqrt{n!}}\,{a^\dagger}^n\psi_0(x), \label{48b}\en
$n\geq0$, in analogy with \cite{bag2022JMAA} we can prove that these functions can be rewritten as:
	\be
	\varphi_n(x)=\frac{1}{\sqrt{n!}}\,\pi_n(x)\varphi_0(x), \qquad \psi_n(x)=\frac{1}{\sqrt{n!}}\,\sigma_n(x) { \psi_0(x)},
	\label{49b}\en
	$n\geq0$, where $\pi_n(x)$ and $\sigma_n(x)$ are defined recursively as follows:
	\be
	\pi_0(x)=\sigma_0(x)=1,
	\label{410b}\en
	and
	\be
	\pi_n(x)=\frac{1}{k}\beta_a(x)\pi_{n-1}(x)-\frac{1}{\beta_a'(x)}\,\pi_{n-1}'(x),
	\label{411b}\en
	\be
	\sigma_n(x)=\overline{\beta_a(x)}\,\sigma_{n-1}(x)-\frac{k}{\overline{\beta_a'(x)}}\,\sigma_{n-1}'(x),
	\label{412b}\en
	$n\geq1$.
Is it further possible to check that these recursive formulas produce the following results, which extend those found in \cite{bag2022JMAA}:
\be
\pi_n(x)=\frac{1}{(2k)^{n/2}}\,H_n\left(\frac{\beta_a(x)}{\sqrt{2k}}\right), \qquad \sigma_n(x)=\left(\frac{k}{2}\right)^{n/2}\,H_n\left(\frac{\overline{\beta_a(x)}}{\sqrt{2k}}\right),
\label{413b}
\en
being $H_n$ the $n$-th Hermite  polynomial.
If we now restrict, for simplicity, to the case of $\alpha(x)>0$ for all $x\in\mathbb{R}$, the functions in (\ref{49b}) can be rewritten as
\be
\varphi_n(x)=\frac{N_\varphi}{\sqrt{n!(2k)^n}}H_n\left(\frac{\beta_a(x)}{\sqrt{2k}}\right) e^{-\left(\frac{\beta_a(x)}{\sqrt{2k}}\right)^2}, 
\label{414b}\en
and
\be
\psi_n(x)=N_\psi\sqrt{\,\frac{1}{n!}\left(\frac{k}{2}\right)^{n}}\,H_n\left(\frac{{\beta_a(x)}}{\sqrt{2k}}\right)\frac{1}{\alpha(x)}.
\label{415b}\en
We observe that we have used here the fact that $\beta_a(x)$ is real, with a clever choice of the integration constant in  (\ref{45b}). Incidentally we also observe that $\beta_a(x)$ is  always increasing, since $\beta_a'(x)=\frac{1}{\alpha(x)}>0$, and $\beta_b(x)$ is also real. The monotone behaviour of $\beta_a(x)$ will play a role later on.

 It is clear that, while $\varphi_n(x)$ can easily be square integrable ({as in all the examples considered in the remaining part of this section}),  $\psi_n(x)$ is not expected to be in $\ltwo$ except than for very special forms of $\alpha(x)$. However, $\psi_n(x)$ and $\varphi_m(x)$ are compatible: $\overline{\psi_n(x)}\varphi_m(x)\in\Lc^1(\mathbb{R})$ for all $n,m=0,1,2,3,\ldots$, and biorthogonal. In particular, if we fix
 \be
 \overline{N_\psi} N_\varphi=\frac{1}{\sqrt{2\pi \,k}},
\label{416b}\en
we have
\be
\langle\psi_n,\varphi_m\rangle=\delta_{n,m},
\label{417b}\en
so that the two families $\F_\varphi=\{\varphi_n(x)\}$ and $\F_\psi=\{\psi_n(x)\}$ are biorthonormal, as in (\ref{33a}). They also satisfy the ladder equations given in (\ref{32a}), and are (generalized) eigenstates of $N=ba$ and $N^\dagger=a^\dagger b^\dagger$, with eigenvalue $n$: $N\varphi_n(x)=n\varphi_n(x)$ and $N^\dagger\psi_n(x)=n\psi_n(x)$. 

After this analysis, we can now deduce the explicit form of the operators $k_\alpha$ and $p_\alpha$ in (\ref{35a}) and (\ref{36a}). We will not consider the operators $l_\alpha$ and $q_\alpha$, since these coincide with the others. Long but straightforward computations produce
the following:
\be\label{418b}
\left\{
\begin{array}{ll}
	k_+=\frac{1}{2}\alpha(x)\left(\alpha(x)\frac{d^2}{dx^2}+\alpha'(x)\frac{d}{dx}\right),\\
	k_-=\frac{1}{2}\left[k^2\alpha^2(x)\frac{d^2}{dx^2}+k\alpha(x)\left(k\alpha'(x)+2\beta_a(x)\right)\frac{d}{dx}+\beta_a^2(x)+k\right],\\
	k_0=-\frac{1}{2}\left[k\alpha^2(x)\frac{d^2}{dx^2}+\alpha(x)\left(\beta_a(x)+k\alpha'(x)\right)\frac{d}{dx} +\frac{1}{2}\right],
\end{array}
\right.
\en
while
\be\label{419b}
\left\{
\begin{array}{ll}
	p_+=\frac{1}{2}[k^2\alpha^2(x)\frac{d^2}{dx^2}+k\alpha(x)\left(3k\alpha'(x)-2\beta_a(x)\right)\frac{d}{dx}+\beta_a^2(x)-k+\\
	\hspace{1cm}+k^2\alpha'^{\,2}(x)+k^2\alpha(x)\alpha''(x)-2k\alpha'(x)\beta_a(x)],\\
	p_-=\frac{1}{2}\left[\alpha^2(x)\frac{d^2}{dx^2}+3\alpha(x)\alpha'(x)\frac{d}{dx}+\alpha(x)\alpha''(x)+\alpha'^{\,2}(x)\right],\\
	p_0=\frac{1}{2}\left[-k\alpha^2(x)\frac{d^2}{dx^2}+\alpha(x)\left(\beta_a(x)-3k\alpha'(x)\right)\frac{d}{dx}+\beta_a(x)\alpha'(x)-k\alpha'^{\,2}(x)-k\alpha(x)\alpha''(x)+\frac{1}{2}\right],
\end{array}
\right.
\en
with also
\be
k^2=p^2=-\frac{3}{16}\1.
\label{420b}\en
In what follows we will consider three different special choices of $\alpha(x)$, and we will deduce what our framework produces in these cases.

\subsection{A special case: constant $\alpha(x)$}\label{subsectionconstantalphas}

We will first consider the case of a constant $\alpha(x)$: $\alpha(x)=\alpha$, a real positive constant. This is a particularly simple, but still absolutely non trivial, situation, as we will show here. First of all, using (\ref{45b}) and the simplest integration constant for $\beta_a(x)$,  we have $\beta_a(x)=\frac{x}{\alpha}$ while $\beta_b(x)=0$. Hence the pseudo-bosonic operators in (\ref{41}) are
$$
a=k\alpha\,\frac{d}{dx}+\frac{x}{\alpha}, \qquad b=-\alpha\,\frac{d}{dx}, 
$$

The operators $k_\alpha$ and $p_\alpha$ turn out to be

\be\label{421b}
\left\{
\begin{array}{ll}
	k_+=\frac{\alpha^2}{2}\frac{d^2}{dx^2},\qquad
	k_-=\frac{1}{2}\left[k^2\alpha^2\frac{d^2}{dx^2}+2kx\frac{d}{dx}+\left(\frac{x^2}{\alpha^2}+k\right)\right],\qquad
	k_0=-\frac{1}{2}\left[k\alpha^2\frac{d^2}{dx^2}+x\frac{d}{dx}+\frac{1}{2}\right],\\
	p_+=\frac{1}{2}\left[k^2\alpha^2\frac{d^2}{dx^2}-2kx\frac{d}{dx}+\left(\frac{x^2}{\alpha^2}-k\right)\right],\qquad p_-=\frac{\alpha^2}{2}\frac{d^2}{dx^2},\qquad 	p_0=\frac{1}{2}\left[-k\alpha^2\frac{d^2}{dx^2}+x\frac{d}{dx}+\frac{1}{2}\right],
\end{array}
\right.
\en
with $k^2$ and $p^2$ as in (\ref{420b}). The functions in (\ref{414b}) and (\ref{415b}) become now
\be
\varphi_n(x)=\frac{N_\varphi}{\sqrt{n!(2k)^n}}H_n\left(\frac{x}{\sqrt{2k}\,\alpha}\right) e^{-\frac{x^2}{2k\alpha^2}}, 
\label{422}\en
and
\be
\psi_n(x)=\frac{N_\psi}{\alpha}\,\sqrt{\frac{1}{n!}\left(\frac{k}{2}\right)^{n}}\,H_n\left(\frac{x}{\sqrt{2k}\,\alpha}\right).
\label{423}\en
It is evident that, since $k>0$, $\varphi_n(x)\in\ltwo$ while $\psi_n(x)\notin\ltwo$, as we have already commented before. And as before, if we take $\overline{N_\psi}\,N_\varphi=\frac{1}{\sqrt{2\pi k}}$, it is easy to check that not only $\overline{\psi_n(x)}\,\varphi_n(x)\in\Lc^1(\mathbb{R})$, but also that 
\be
\langle\psi_n,\varphi_m\rangle=\delta_{n,m}.
\label{424}\en
So the families $\F_\varphi$ and $\F_\psi$ are biorthonormal (in this extended sense). The vectors in (\ref{39a}), (\ref{314}), (\ref{319}) and (\ref{323}) can be easily found:
\be\label{425}
\left\{
\begin{array}{ll}
	\varphi_{-\frac{1}{4},m+\frac{1}{4}}(x)=\frac{N_\varphi}{(2m-1)!!(2k)^m}H_{2m}\left(\frac{x}{\sqrt{2k}\,\alpha}\right) e^{-\frac{x^2}{2k\alpha^2}},\\
	\psi_{-\frac{1}{4},m+\frac{1}{4}}(x)=\frac{N_\psi (2m-1)!!}{\alpha(2m)! }\,\left(\frac{k}{2}\right)^{m}\,H_{2m}\left(\frac{x}{\sqrt{2k}\,\alpha}\right), \\
	\varphi_{-\frac{1}{4},m+\frac{3}{4}}(x)=\frac{N_\varphi}{(2m-1)!!(2k)^{m+1/2}}H_{2m+1}\left(\frac{x}{\sqrt{2k}\,\alpha}\right) e^{-\frac{x^2}{2k\alpha^2}},\\
	\psi_{-\frac{1}{4},m+\frac{3}{4}}(x)=\frac{N_\psi (2m-1)!!}{\alpha(2m+1)! }\,\left(\frac{k}{2}\right)^{m+1/2}\,H_{2m+1}\left(\frac{x}{\sqrt{2k}\,\alpha}\right). \\
\end{array}
\right.
\en

{In Figures \ref{fig:ex1} we plot some of these functions with the following choice of parameters: $m=1$, $k=2$,  and $\alpha=3$. The critical behaviours of the $\psi$-functions (see the blue lines) are evident (they are clearly not in $\Lc^2(\mathbb{R})$), while the $\varphi$-functions (orange lines) are in $\Lc^2(\mathbb{R})$. Moreover, the products between the  $\varphi$-functions and $\psi$-functions (black dotted lines), due to their compatibility, are in $\Lc^1(\mathbb{R})$, as expected. These claims are all visible in the large $x$ behaviour of the functions in figure (even if the range of the $x$ variable is not so extended).}

\begin{figure}[!htbp]
	
	\begin{center}
		\subfigure[]{\includegraphics[width=8cm]{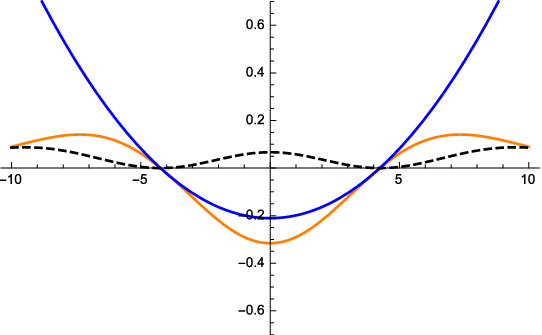}}
	\subfigure[]{\includegraphics[width=8cm]{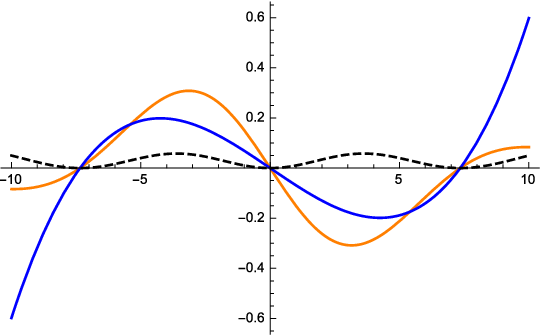}}
	
	\caption{(a) The plots of $\varphi_{-\frac{1}{4},\frac{5}{4}}(x)$ (orange-continuous line), $\psi_{-\frac{1}{4},\frac{5}{4}}(x)$ (blue-continuous line), $\varphi_{-\frac{1}{4},\frac{5}{4}}(x)\psi_{-\frac{1}{4},\frac{5}{4}}(x)$ (black-dashed line). 
	(b) The plots of $\varphi_{-\frac{1}{4},\frac{7}{4}}(x)$ (orange-continuous line), $\psi_{-\frac{1}{4},\frac{7}{4}}(x)$ (blue-continuous line), $\varphi_{-\frac{1}{4},\frac{7}{4}}(x)\psi_{-\frac{1}{4},\frac{7}{4}}(x)$ (black-dashed line).  In all plots  we set in \eqref{425} $m=1$, $k=2$, $N_\varphi=N_\psi=\frac{1}{\sqrt{2}\,\pi^{1/4}}$,  and $\alpha=3$.  }\label{fig:ex1}
	\end{center}

\end{figure}

It is now easy to identify the sets $\F_\varphi^{(e)}\left(\frac{1}{4}\right)$, $\F_\varphi^{(o)}\left(\frac{1}{4}\right)$, $\F_\psi^{(e)}\left(\frac{1}{4}\right)$ and $\F_\psi^{(o)}\left(\frac{1}{4}\right)$ introduced before, as well as the spaces $\Hil_{-\frac{1}{4}}^{(e,o)}$. However, the various $\psi_{-\frac{1}{4},m+\frac{3}{4}}(x)$ do not belong to $\Hil_{-\frac{1}{4}}^{(e,o)}$. Still, they are biorthonormal to the various $\varphi_{-\frac{1}{4},m+\frac{3}{4}}(x)$. Moreover, $(\F_\varphi,\F_\psi)$ are $\E$-quasi bases, see below. This can be restated by saying that $(\F_\varphi^{(e)}\left(\frac{1}{4}\right)\cup\F_\varphi^{(o)}\left(\frac{1}{4}\right),\F_\psi^{(e)}\left(\frac{1}{4}\right)\cup\F_\psi^{(o)}\left(\frac{1}{4}\right))$ are $\E$-quasi bases, which means, \cite{bagspringer}, that, for all $f(x), g(x)\in\E$, the following is true:
\be
\sum_{n=0}^\infty \langle f,\varphi_n\rangle\langle \psi_n,g\rangle=\sum_{n=0}^\infty \langle f,\psi_n\rangle\langle \varphi_n,g\rangle=\langle f,g\rangle.
\label{427}\en

The set $\E$ is defined as follows:
\be
\E=\left\{h(x)\in\ltwo:\, h_-(x):=h(\sqrt{2k}\,\alpha\,x))\,e^{x^2/2}\in\ltwo\right\} 
\label{426}\en
This set is dense in $\ltwo$. Indeed, it contains the set $\D(\mathbb{R})$ of all the compactly supported $C^\infty$ functions. In the following we will also need the following: if $h(x)\in\E$, then the function $h_+(x):=h(\sqrt{2k}\,\alpha\,x))\,e^{-x^2/2}\in\ltwo$ as well, as it is clear. 
To prove now (\ref{427}) we first observe that, if $f(x), g(x)\in\E$, with simple changes of variables we have
\be
\langle\psi_n,g\rangle=\overline N_\psi k^{n/2}\sqrt{2k\sqrt{\pi}}\langle e_n,g_-\rangle, \qquad \langle f,\varphi_n\rangle=N_\varphi \frac{\alpha}{k^{n/2}}\,\sqrt{2k\sqrt{\pi}}\langle f_+,e_n,\rangle,
\label{428}\en
for all $n\geq0$, { and where $g_-$ and $f_+$ are defined as above}. Here $e_n(x)$=$\frac{1}{\sqrt{2^n\,n!\sqrt{\pi}}}H_n(x)e^{-x^2/2}$ is the $n$-th eigenfunction of the quantum harmonic oscillator. Since $\F_e=\{e_n(x), n\geq0\}$ is an orthonormal basis { in}  $\ltwo$ we have, with simple computations
$$
\sum_{n=0}^\infty \langle f,\varphi_n\rangle\langle \psi_n,g\rangle=N_\varphi\overline N_\psi 2k\alpha\sqrt{\pi} \sum_{n=0}^\infty \langle f_+,e_n\rangle\langle e_n,g_-\rangle=N_\varphi\overline N_\psi 2k\alpha\sqrt{\pi}  \langle f_+,g_-\rangle.
$$ 
Next, since $\langle f_+,g_-\rangle=\frac{1}{\sqrt{2k}\,\alpha}\langle f,g\rangle$  and $\overline{N_\psi}\,N_\varphi=\frac{1}{\sqrt{2\pi k}}$, half of (\ref{427}) is proven. The other half can be proved similarly.

\vspace{2mm}

{\bf Remarks:--} (1) first of all we stress once more that, despite of the fact that the $\psi$ functions are not square integrable, they still provide, together with the $\varphi$'s, a resolution of the identity at least on a dense set, $\E$. This is in agreement with what has been widely discussed recently for {\em weak pseudo-bosons}, \cite{bagspringer}.

(2) Formulas (\ref{425}), together with the parity properties of the Hermite polynomials, show how to split  $\ltwo$ in two {\em orthogonal sectors}. In particular, if we consider the following two linear spans of the $e_n(x)$, $\E_e=l.s.\{e_{2n}(x), n\geq0\}$ and  $\E_o=l.s.\{e_{2n+1}(x), n\geq0\}$, it is clear that
$$
\langle\varphi_{-\frac{1}{4},m+\frac{1}{4}},f\rangle=\langle\psi_{-\frac{1}{4},m+\frac{1}{4}},f\rangle=0,
$$
for all $f(x)\in\E_o$, and that 
$$
\langle\varphi_{-\frac{1}{4},m+\frac{3}{4}},g\rangle=\langle\psi_{-\frac{1}{4},m+\frac{3}{4}},g\rangle=0,
$$
for all $g(x)\in\E_e$. Notice that the equalities above involving  $\psi_{-\frac{1}{4},m+\frac{1}{4}}$ and $\psi_{-\frac{1}{4},m+\frac{3}{4}}$ are guaranteed by the fact that, as it is easy to see, both $\psi_{-\frac{1}{4},m+\frac{1}{4}}f(x)$ and $\psi_{-\frac{1}{4},m+\frac{3}{4}} g(x)$ belong to $\Lc^1(\mathbb{R})$.

(3) As a last remark, we observe that, if $k_\alpha$ and $p_\alpha$ are those in (\ref{421b}), the functions in (\ref{425}) satisfy, among the others, the equalities
$$
k_-\varphi_{-\frac{1}{4},\frac{1}{4}}(x)=0, \quad k_+\varphi_{-\frac{1}{4},\frac{1}{4}}(x)=\frac{1}{2}\,\varphi_{-\frac{1}{4},\frac{5}{4}}(x), \quad k_+^2\varphi_{-\frac{1}{4},\frac{1}{4}}(x)=\frac{3}{4}\,\varphi_{-\frac{1}{4},\frac{9}{4}}(x), \ldots, 
$$
with $k_0\varphi_{-\frac{1}{4},m+\frac{1}{4}}(x)=\left(m+\frac{1}{4}\right)\varphi_{-\frac{1}{4},m+\frac{1}{4}}(x)$, and similar equations for $p_\alpha\psi_{-\frac{1}{4},m+\frac{1}{4}}(x)$. It is maybe useful to remind that the connection between, say, $\varphi_{-\frac{1}{4},m+\frac{1}{4}}(x)$ and $\varphi_{-\frac{1}{4},m+\frac{3}{4}}(x)$ is provided by $b$, see (\ref{41}), and not by the operators of $\mathfrak{su}(1,1)$. For instance,  $b\varphi_{-\frac{1}{4},m+\frac{1}{4}}(x)=\varphi_{-\frac{1}{4},m+\frac{3}{4}}(x)$.

\subsection{A first not constant $\alpha(x)$}\label{subsection-x4}

In this second example we put $\alpha(x)=\dfrac{1}{1+ \gamma x^4}$, a real function,  with $\gamma$ a positive constant. As required, $\alpha(x)$ is strictly positive. Then  $\beta_a(x)= x + \dfrac{\gamma x^5}{5}$. The pseudo-bosonic operators in (\ref{41}) are
$$
a=k\dfrac{1}{1+ \gamma x^4}\,\frac{d}{dx}+x + \dfrac{\gamma x^5}{5}, \qquad b=-\dfrac{1}{1+ \gamma x^4}\frac{d}{dx}, 
$$
From equations
\eqref{418b}--\eqref{420b} we get, for instance $k_+= \dfrac{1}{2} \left[ \dfrac{1}{(1+ \gamma x^4)^2} \dfrac{d^2}{dx^2} -  \dfrac{4 \gamma x^3}{(1+ \gamma x^4)^3}  \dfrac{d}{dx}  \right]$,  and the vacua of $a$ and $b^\dagger$ are:
\be
\varphi_0(x)=N_\varphi e^{-\frac{1}{2k}\, \left( x + \frac{\gamma x^5}{5}\right)^2}, \qquad \psi_0(x)=N_\psi \left(1+ \gamma x^4 \right).
\label{47c}\en
Furthermore, the functions in (\ref{414b}) and (\ref{415b}) become now:
\be 
\varphi_n(x)=\frac{N_\varphi}{\sqrt{n!(2k)^n}}H_n\left(\frac{x+ \frac{ \gamma \, x^5}{5}}{\sqrt{2k}}\right) e^{-\frac{1}{2k}\left(x + \frac{\gamma \,  x^5}{5} \right)^2}, 
\label{422b}\en
and
\be
\psi_n(x)= N_\psi\,\sqrt{\frac{1}{n!}\,\left(\frac{k}{2}\right)^{n}}\,H_n\left(\frac{x+ \frac{ \gamma \, x^5}{5}}{\sqrt{2k}}\right) \left( 1 + \gamma x^4 \right).
\label{423b}\en

{\bf Remark:--}   It is interesting to observe that, when $\gamma$ goes to zero, $\alpha(x)=1$, and we recover exactly the functions in \eqref{422} and \eqref{423} with  $\alpha=1$.

\vspace{0,3 cm}

As in the previous case, since $k>0$, $\varphi_n(x)\in\ltwo$ while $\psi_n(x)\notin\ltwo$. And as before, if we take $\overline{N_\psi}\,N_\varphi=\frac{1}{\sqrt{2\pi k}}$, it is easy to check that $\overline{\psi_n(x)}\,\varphi_n(x)\in\Lc^1(\mathbb{R})$, and \eqref{424} is still valid.

So, in analogy with the previous subsection, we consider the families $\F_\varphi$ and $\F_\psi$, and the respective vectors become (as in (\ref{39a}), (\ref{314}), (\ref{319}) and (\ref{323})):
\be\label{425b}
\left\{
\begin{array}{ll}
	\varphi_{-\frac{1}{4},m+\frac{1}{4}}(x)=\frac{N_\varphi}{(2m-1)!!(2k)^m}H_{2m}\left(\frac{x+ \frac{ \gamma \, x^5}{5}}{\sqrt{2k}}\right) e^{-\frac{1}{2k}\left(x + \frac{\gamma \,  x^5}{5} \right)^2}\\
	\psi_{-\frac{1}{4},m+\frac{1}{4}}(x)=\frac{N_\psi (2m-1)!!}{(2m)! }\,\left(\frac{k}{2}\right)^{m}\,H_{2m}\left(\frac{x+ \frac{ \gamma \, x^5}{5}}{\sqrt{2k}}\right) \left( 1 + \gamma x^4 \right), \\
	\varphi_{-\frac{1}{4},m+\frac{3}{4}}(x)=\frac{N_\varphi}{(2m-1)!!(2k)^{m+1/2}}H_{2m+1}\left(\frac{x+ \frac{ \gamma \, x^5}{5}}{\sqrt{2k}}\right) e^{-\frac{1}{2k}\left(x + \frac{\gamma \,  x^5}{5} \right)^2},\\
	\psi_{-\frac{1}{4},m+\frac{3}{4}}(x)=\frac{N_\psi (2m-1)!!}{(2m+1)! }\,\left(\frac{k}{2}\right)^{m+1/2}\,H_{2m+1}\left(\frac{x+ \frac{ \gamma \, x^5}{5}}{\sqrt{2k}}\right) \left( 1 + \gamma x^4 \right). \\
\end{array}
\right.
\en

In Figure \ref{fig:ex2} we plot some of these functions with the following choice of parameters: $m=1$,  $k=1/2$, $N_\varphi=N_\psi=\frac{1}{\pi^{1/4}}$, and $\gamma=1/2$. Once again, we observe the crucial behaviors exhibited by the $\psi$-functions, whereas   compatibility between the $\varphi$-functions and the $\psi$-functions is maintained, as suggested by the {\em large} $x$ behaviour\footnote{The reason why we restrict here to $x\in[-3,3]$ is due to the fact that, otherwise, $\varphi_{-\frac{1}{4},\frac{5}{4}}(x)$ is not really visible in the plot, since $\psi_{-\frac{1}{4},\frac{5}{4}}(x)$ diverges very fast to $+\infty$.} of the functions in the plots.
\begin{figure}[!htbp]	
	\begin{center}
		\subfigure[]{\includegraphics[width=8cm]{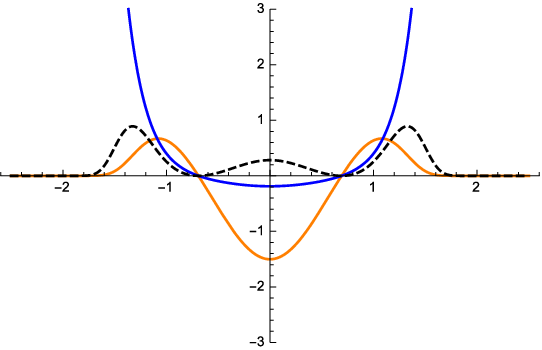}}
		\subfigure[]{\includegraphics[width=8cm]{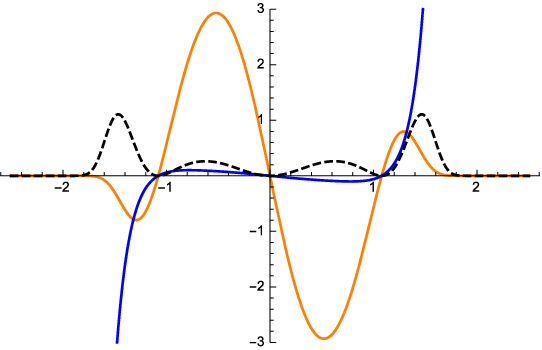}}
		\caption{(a) The plots of $\varphi_{-\frac{1}{4},\frac{5}{4}}(x)$ (orange-continuous line), $\psi_{-\frac{1}{4},\frac{5}{4}}(x)$ (blue-continuous line), $\varphi_{-\frac{1}{4},\frac{5}{4}}(x)\psi_{-\frac{1}{4},\frac{5}{4}}(x)$ (black-dashed line) 
			(b) The plots of $\varphi_{-\frac{1}{4},\frac{7}{4}}(x)$ (orange-continuous line), $\psi_{-\frac{1}{4},\frac{7}{4}}(x)$ (blue-continuous line), $\varphi_{-\frac{1}{4},\frac{7}{4}}(x)\psi_{-\frac{1}{4},\frac{7}{4}}(x)$ (black-dashed line). In all plots  we set in  \eqref{425b}   $m=1$,  $k=1/2$, $N_\varphi=N_\psi=\frac{1}{\pi^{1/4}}$, and $\gamma=1/2$.} \label{fig:ex2}
	\end{center}

\end{figure}

As in the previous case it is easy to identify the sets $\F_\varphi^{(e)}\left(\frac{1}{4}\right)$, $\F_\varphi^{(o)}\left(\frac{1}{4}\right)$, $\F_\psi^{(e)}\left(\frac{1}{4}\right)$ and $\F_\psi^{(o)}\left(\frac{1}{4}\right)$, so that  $(\F_\varphi,\F_\psi)$ are $\E$-quasi bases, when the set  $\E$ is defined as follow:

\be
\E=\left\{h(x)\in\ltwo:\, h_-(x):=h\left(\beta_a^{-1}\left(\sqrt{2 k} x\right)\right) \,e^{x^2/2}\in\ltwo\right\},
\label{426b}\en
where, in this particular case, $\beta_a^{-1}\left(\sqrt{2 k} x\right)= \sqrt{2 k} \, x \, _4F_3\left(\frac{1}{5},\frac{2}{5},\frac{3}{5},\frac{4}{5};\frac{1}{2},\frac{3}{4},\frac{5}{4};-\frac{625}{64} k^2 x^4 \gamma \right)$. Here $ _4F_3(a_i,b_j,z)$ is the generalized Hypergeometric function, that converges for $|z|< 1$, i.e. if $x^2<\frac{8}{25\,\sqrt{\gamma}\,k}$. 

The set $\E$ is again dense in $\ltwo$. Indeed, it contains the set $\D(\mathbb{R})$ of all the compactly supported $C^\infty$ functions.

It is possible to check that  $(\F_\varphi,\F_\psi)$ are $\E$-quasi bases.
As in the previous example we can still define  $h_+(x):=h\left(\beta_a^{-1}\left(\sqrt{2 k} x\right)\right)\alpha\left(\beta_a^{-1}\left(\sqrt{2 k} x\right)\right)\, e^{-x^2/2}$. The difference with our previous definition of $h_+(x)$ is due to the fact that, in this case, $\alpha(x)$ is not constant. It is clear that $h_+(x)$ $\in\ltwo$. 

To prove next the equality in (\ref{427}) we now observe that, if $f(x), g(x)\in\E$, with simple changes of variables we have
\be
\langle\psi_n,g\rangle=\overline N_\psi k^{n/2}\sqrt{2k\sqrt{\pi}}\langle e_n,g_-\rangle, \qquad \langle f,\varphi_n\rangle=N_\varphi \frac{\sqrt{2k\sqrt{\pi}}}{k^{n/2}}\,\langle f_+,e_n,\rangle,
\label{428b}\en
for all $n\geq0$. These expressions are almost like those in \eqref{428}, except for the term $\alpha$ in the second expression, which, in this case, is incorporated in the definition of $f_+$, while $e_n(x)$ is still  the $n$-th eigenfunction of the quantum harmonic oscillator. 
Since $\F_e=\{e_n(x), n\geq0\}$ is an orthonormal basis in $\ltwo$ we have, with simple computations
$$
\sum_{n=0}^\infty \langle f,\varphi_n\rangle\langle \psi_n,g\rangle=N_\varphi\overline N_\psi 2k \sqrt{\pi} \sum_{n=0}^\infty \langle f_+,e_n\rangle\langle e_n,g_-\rangle=N_\varphi\overline N_\psi 2k \sqrt{\pi}  \langle f_+,g_-\rangle.
$$ 
Next, since $\langle f_+,g_-\rangle=\frac{1}{\sqrt{2k}}\langle f,g\rangle$  and $\overline{N_\psi}\,N_\varphi=\frac{1}{\sqrt{2\pi k}}$, half of (\ref{427}) is proven. The other half can be proved similarly.

\subsection{A third example}\label{subsection-cosx}

Next we consider the function $\alpha(x)=\dfrac{1}{1+ \gamma \cos x}$. Also in this case the function $\alpha(x)$ is a real function, which is positive if we restrict the constant  $\gamma$ to assume values in $]-1,1[$. Notice that this $\alpha(x)$ is oscillating, while those we have considered before are not. 
We have  $\beta_a(x)= x + \gamma \sin x $. 
The pseudo-bosonic operators in (\ref{41}) are
$$
a=k\dfrac{1}{1+ \gamma \cos x}\,\frac{d}{dx}+x + \gamma \sin x, \qquad b=-\dfrac{1}{1+ \gamma \cos x}\frac{d}{dx}, 
$$
and
\eqref{418b}--\eqref{420b} return, for instance,  $k_+=\dfrac{1}{2} \left[ \dfrac{1}{(1+ \gamma \cos x)^2} \dfrac{d^2}{dx^2} +  \dfrac{\gamma \sin x}{(1+ \gamma  \cos x)^3}  \dfrac{d}{dx}  \right]$. The vacua of $a$ and $b^\dagger$ are:
\be
\varphi_0(x)=N_\varphi e^{-\frac{1}{2k}\, \left( x + \gamma \sin x \right)^2}, \qquad \psi_0(x)=N_\psi \left(1 + \gamma \cos x \right).
\label{47d}\en \\
The functions in (\ref{414b}) and (\ref{415b}) are now:
\be 
\varphi_n(x)=\frac{N_\varphi}{\sqrt{n!(2k)^n}}H_n\left(\frac{x+  \gamma \sin x}{\sqrt{2k}}\right) e^{-\frac{1}{2k}\left(x +  \gamma \sin x \right)^2}, 
\label{422c}\en
and
\be
\psi_n(x)=N_\psi\,\sqrt{\frac{1}{n!}\left(\frac{k}{2}\right)^{n}}\,H_n\left(\frac{x+  \gamma \sin x}{\sqrt{2k}}\right) \left( 1 + \gamma \cos x \right)
\label{423c}\en
which can also be rewritten as follows:
$$
\psi_n(x)= \frac{N_\psi}{n+1}\,\sqrt{\frac{2k}{n!}\,\left(\frac{k}{2}\right)^{n}}\, H'_{n+1} \left(\frac{x+  \gamma \sin x}{\sqrt{2k}}\right).
$$

It is interesting to stress that, once more, when $\gamma $  tends to zero, we obtain exactly the functions in \eqref{422} and in \eqref{423} with $\alpha=1$. Furthermore, since $k>0$, $\varphi_n(x)\in\ltwo$ while $\psi_n(x)\notin\ltwo$.

In this case, the vectors in (\ref{39a}), (\ref{314}), (\ref{319}) and (\ref{323}) are:
\be\label{425c}
\left\{
\begin{array}{ll}
	\varphi_{-\frac{1}{4},m+\frac{1}{4}}(x)=\frac{N_\varphi}{(2m-1)!!(2k)^m}H_{2m}\left(\frac{x+  \gamma \sin x}{\sqrt{2k}}\right) e^{-\frac{1}{2k}\left(x +  \gamma \sin x \right)^2},\\
	\psi_{-\frac{1}{4},m+\frac{1}{4}}(x)=\frac{N_\psi (2m-1)!!}{(2m)! }\,\left(\frac{k}{2}\right)^{m}\,H_{2m}\left(\frac{x+  \gamma \sin x}{\sqrt{2k}}\right) \left( 1 + \gamma \cos x \right), \\
	\varphi_{-\frac{1}{4},m+\frac{3}{4}}(x)=\frac{N_\varphi}{(2m-1)!!(2k)^{m+1/2}}H_{2m+1}\left(\frac{x+  \gamma \sin x}{\sqrt{2k}}\right) e^{-\frac{1}{2k}\left(x +  \gamma \sin x \right)^2},\\
	\psi_{-\frac{1}{4},m+\frac{3}{4}}(x)=\frac{N_\psi (2m-1)!!}{(2m+1)! }\,\left(\frac{k}{2}\right)^{m+1/2}\,H_{2m+1}\left(\frac{x+  \gamma \sin x}{\sqrt{2k}}\right) \left( 1 + \gamma \cos x \right). \\
\end{array}
\right.
\en
The plots of some of these functions are shown in Figure \ref{fig:ex3} with the following choice of parameters: $m=1$, $k=\frac{1}{2}$,  $N_\varphi=N_\psi=\frac{1}{\pi^{1/4}}$, and $\gamma=1/2$. The crucial behaviors of the $\psi$-functions is particularly evident in panels (c)-(d), where we observe their oscillating and diverging characteristics as $|x|\rightarrow\infty$. These oscillations show the relevance of the explicit choice of $\alpha(x)$.

\begin{figure}[!htbp]
	\begin{center}
		\subfigure[]{\includegraphics[width=8cm]{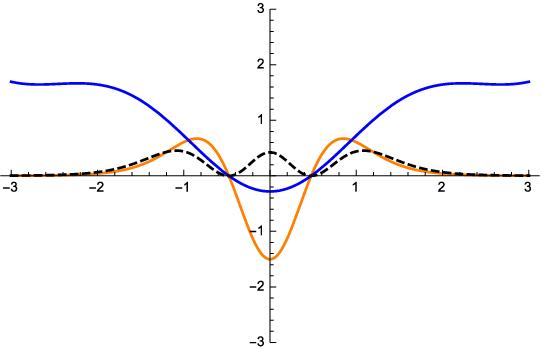}}
		\subfigure[]{\includegraphics[width=8cm]{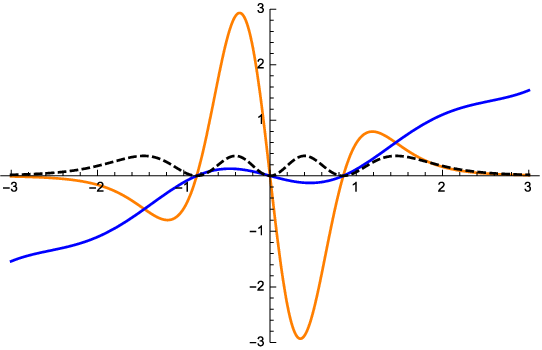}}\\
		\subfigure[]{\includegraphics[width=7cm]{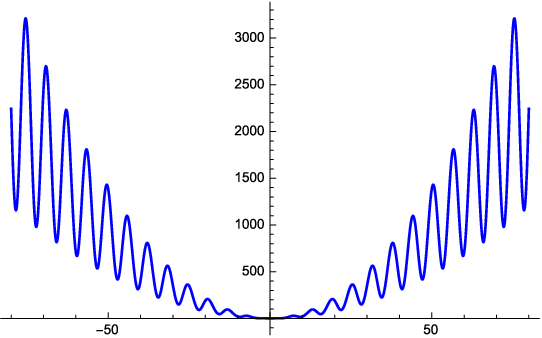}}
		\hspace*{0.5cm}\subfigure[]{\includegraphics[width=7cm]{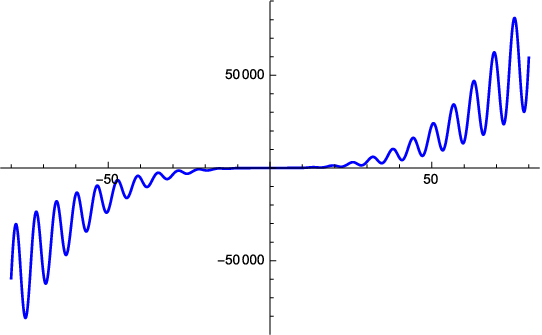}}
		\caption{(a) The plots of $\varphi_{-\frac{1}{4},\frac{5}{4}}(x)$ (orange-continuous line), $\psi_{-\frac{1}{4},\frac{5}{4}}(x)$ (blue-continuous line), $\varphi_{-\frac{1}{4},\frac{5}{4}}(x)\,\psi_{-\frac{1}{4},\frac{5}{4}}(x)$ (black-dashed line).
			(b) The plots of $\varphi_{-\frac{1}{4},\frac{7}{4}}(x)$ (orange-continuous line), $\psi_{-\frac{1}{4},\frac{7}{4}}(x)$ (blue-continuous line), $\varphi_{-\frac{1}{4},\frac{7}{4}}(x)\psi_{-\frac{1}{4},\frac{7}{4}}(x)$ (black-dashed line). (c) The plot of $\psi_{-\frac{1}{4},\frac{5}{4}}(x)$ as in (a) in a larger domain. (d) The plot of $\psi_{-\frac{1}{4},\frac{7}{4}}(x)$ as in (b) in a larger domain.  In all plots  we set in \eqref{425c}  $m=1$, $k=\frac{1}{2}$,  $N_\varphi=N_\psi=\frac{1}{\pi^{1/4}}$, and $\gamma=1/2$. }	\label{fig:ex3}
	\end{center}

\end{figure}

Since $-1 < \gamma < 1$, $\alpha(x)$ is always positive and, also in this case, the function $\beta_a(x)$ is invertible and we can define the set $\E$ so that  $(\F_\varphi,\F_\psi)$ are $\E$-quasi bases, as in \eqref{426b}. However,in this case it is not easy to find an explicit form of the inverse of $\beta_a(x)$. Nevertheless, all the main results, which are clearly model-independent, hold true.

\section{Some results on squeezed states}\label{sectsquizi}
In this section, we explore the possibility of defining a weak formulation for squeezed states, which are states of great significance in physics. 
As widely known, squeezed states in Quantum Mechanics can be obtained via the action of the unitary squeezing operator on the ground of the harmonic oscillator. Another way of defining a squeezed state is through the definition of a Hamiltonian operator that naturally arises, in our context, when working with the triplet of operators $k_0, k_-, k_+$ defined in \eqref{35a}. We can in fact define the following, not self-adjoint, Hamiltonian:

\bea
H=2\mu(z)\,k_0 +2\lambda(z)\,k_-+2\lambda(\bar z)\,k_{+}+\left(\sinh(r)-\frac{\mu(z)}{2}\right)\1,
\ena
where $z=re^{i\theta},\,\mu(z)=\cosh(2r),\,\lambda(z)=e^{-i\theta}\cosh(r)\sinh(r)$.
Expressed in terms of the operators $a$ and $b$ (as defined in equation \eqref{41}), the Hamiltonian can be written as
\bea
H=\mu(z)ba +\lambda(z)a^2+\lambda(\bar z)\,b^2+\sinh(r)\1,
\ena
which even more expresses that $H$ is manifestly non self-adjoint, $H\neq H^\dagger$, due to the pseudo-bosonic nature of $a$ and $b$.
Next, if we introduce the operators
\bea
A=\cosh(r)a+e^{i\theta}\sinh(r)b &=&\left(\cosh(r)-e^{i\theta}\sinh(r)\right)\alpha(x)\frac{d}{dx}+\cosh(r)\beta_a(x),\label{add1}\\ 
B=\cosh(r)b+e^{-i\theta}\sinh(r)a &=&\left(e^{-i\theta}\sinh(r)-\cosh(r)\right)\alpha(x)\frac{d}{dx}+e^{-i\theta}\sinh{(r)}\beta_a(x),\nonumber\\
\label{add2}\ena
{where we have used  \eqref{45d} with the particular choice $k=1$}, they  satisfy (on a suitable function space, see below), $[A,B]=\1$, so that $A$ and $B$ can also be seen as pseudo-bosonic operators, as their counterparts lower letter counterparts $a$ and $b$, and
$H$ can be factorized as
\bea
H=BA.
\ena
The relationship between $H$ and a generalization of squeezed states becomes evident through the following arguments. Let us assume, for a moment, that $b=a^\dagger$. By performing standard computations similar to those presented in \cite{bisquizi}, it can be shown that $H=\mathcal{S}(z) b a\mathcal{S}(z)^{-1}$, where $\mathcal{S}(z)=e^{\frac{1}{2}zb^2-\frac{1}{2}\bar{z}a^2},\, z\in \mathbb{C}$, represents the unitary squeezing operator. Consequently, the ground state $\tau(z)$ of $H$, such that $H\tau(z)=0$, corresponds exactly to a squeezed state, as discussed in \cite{bisquizi,BGR19}.
Refocusing on the case where $a$ and $b$ are defined as in equation \eqref{41}, it is now natural to define a squeezed state as the ground state of $H$ or, more in general, as the eigenstate relative to the null eigenvalue. We observe that the factorization $H=BA$ implies that this ground state, whether it exists or can be defined in a weak sense, can be found by requiring that it is  annihilated by $A$. Similarly, we can suppose the existence of a squeezed state that is a ground of $H^\dagger=A^\dagger B^\dagger$ and it is annihilated by  $B^\dagger$.
It is evident that the only possibility for the existence of this pair of squeezed states is to define them in a distributional sense, treating them as proper functionals on a suitable set of functions. The reason for this is that, due to the critical behavior of the two sets $\F_\Psi$ and $\F_\varphi$, it is not recommended, and perhaps impossible, to try to define the squeezed states as series, convergent in $\ltwo$ for all $z\in \mathbb{C}$ (or in some domain of convergence), of the form
$$
\tau(z)=\sum_{n\geq0}\nu_{\varphi}(z,n)\varphi_n,\quad \kappa(z)=\sum_{n\geq0}\nu_{\psi}(z,n)\psi_n,
$$
for some suitable choice of $\nu_{\varphi}(z,n),\nu_{\psi}(z,n)$.
Also, the squeezing operator becomes unbounded, and we should pay attention to, just to cite one problem, its domain. Given that, let us introduce the candidate squeezed states $\tau(z)$ and $\kappa(z)$ as functionals 
$\in \E_c^{'}$ where
\bea \E_c=\left\{h(x) \in \mathcal{L}^2(\mathbb{R}): h_{-}(x):=h\left(\beta_a^{-1}(\sqrt{2} x)\right) e^{x^2 / 2} \in \mathcal{L}^2(\mathbb{R})\right\}.\ena  
The above set  is a generalization of \eqref{426} and its properties where already discussed in \cite{bag2022JMAA}. We observe that the density of this set in $\ltwo$ has already been discussed previously in this paper. We have also discussed that, under very mild conditions on $\alpha(x)$, if $h(x) \in \mathcal{E}_c$, then the function $h_{+}(x):=h\left(\beta_a^{-1}(\sqrt{2} x)\right) \alpha\left(\beta_a^{-1}(\sqrt{2} x)\right) e^{-x^2 / 2} \in \mathcal{L}^2(\mathbb{R})$ as well. Moreover, for $h(x)\in\E_c$ we have
\bea
\left\langle h, \varphi_n\right\rangle=N_{\varphi} \pi^{1 / 4} \sqrt{2}\left\langle h_{+}, e_n\right\rangle,\quad \left\langle \psi_n ,h\right\rangle=\overline{N}_{\psi} \pi^{1 / 4} \sqrt{2}\left\langle  e_n, h_{-}\right\rangle\label{45c}
\ena
where $e_n(x)$, the $n$-th eigenstate of the quantum harmonic oscillator, was introduced before, and where $N_{\varphi},N_{\psi}$ are normalization factor not particularly important here.
We now define $\tau(z),\kappa(z)$ via the following functional actions on $g(x)\in\E_c$
\bea
F_{\tau}[z](g)&=&\langle \tau(z),g\rangle=e^{\nu(\bar z)}\sum_{k\geq 0}\frac{\eta(\bar z)^k}{k!}\sqrt{(2k)!}\,\langle \varphi_{2k},g\rangle,\label{46ba}\\
F_{\kappa}[z](g)&=&\langle \kappa(z),g\rangle=e^{\nu(\bar z)}\sum_{k\geq 0}\frac{\eta(\bar z)^k}{k!}\sqrt{(2k)!}\,\langle \psi_{2k},g\rangle,\label{46bb}
\ena 
where $\nu(z)=-\frac{1}{2}\log(\cosh(r))$ and $\eta(z)=-\frac{1}{2}e^{i\theta}\tanh(r)$.
Is is possible to check that the above series are well defined for all $z\in\mathbb{C}$. This simply follows from the relations 
\bea
|\left\langle g, \varphi_n\right\rangle|\leq |N_{\varphi}| \pi^{1 / 4} \sqrt{2}\,\|g_{+}\|,\quad |\left\langle g, \psi_n\right\rangle|\leq |N_{\psi}| \pi^{1 / 4} \sqrt{2}\,\|g_{-}\|,
\ena
by the fact that $g_{+},g_-\in L^2(\mathbb{R})$, and from a straightforward 
inspection of the radii of convergence of the series.
We can also check that $F_{\tau}[z],F_{\kappa}[z]\in\E_c^{'}$. In fact, it is evident that they are linear. Moreover, adopting similar arguments as those used in \cite{bag2022JMAA}, they are also $\tau_{\E_c}$-continuous being $\tau_{\E_c}$ the following topology in $\E_c$:
we say that a sequence $\left\{g_n(x)\right\}$ in $\mathcal{E}_c$ is $\tau_{\mathcal{E}_c}$-convergent to a certain $g(x) \in \mathcal{L}^2(\mathbb{R})$ if $\left\{g_n(x)\right\}$, $\left\{\left(g_n\right)_{+}(x)\right\}$ and $\left\{\left(g_n\right)_{-}(x)\right\}$ converge to $g(x)$, $g_{+}(x)$ and  $g_{-}(x)$, respectively, in the norm $\|\cdot\|$ of $\mathcal{L}^2(\mathbb{R})$. It is clear that, when this is true, $g(x) \in \mathcal{E}_c$, and hence $\mathcal{E}_c$ is closed in $\tau_{\mathcal{E}_c}$.
To prove the continuity of the functionals we simply consider
\bea
|F_{\tau}[z](g_n-g)|\leq &=&|e^{\nu(z)}|\left(\sum_{k\geq 0}\frac{|\eta( z)|^k}{k!}\sqrt{(2k)!}\right)\,\|(g_n)_{+}-g_+ \|\rightarrow 0,\\
|F_{\kappa}[z](g_n-g)|\leq& =&|e^{\nu(z)}|\left(\sum_{k\geq 0}\frac{|\eta( z)|^k}{k!}\sqrt{(2k)!}\right)\,\|(g_n)_{-}-g_- \|\rightarrow 0,
\ena 
which is true for all $z\in \mathbb{C}$, 
and hence $\tau(z)\in \E_c^{'}$.\\
 {We are now ready to  show that $\tau(z)$ and $\kappa(z)$ are annihilated (in a {\em distributional way}, see below) by $A$ and $B^\dagger$, respectively, suggesting the possibility to define them as generalized (weak) squeezed states.} This is the content of the following proposition:
\begin{prop}
	The pair $\tau(z),\kappa(z)$ satisfies the following properties:
	\bea
	\langle g,A\tau(z) \rangle=\langle g,B^\dagger\kappa(z) \rangle=0,\, \forall g(x)\in\D(\mathbb{R}).
	\ena
\end{prop}
\begin{proof}
We start proving that
\bea
\langle g,A\,\tau(z) \rangle=\langle A^\dag g,\tau(z) \rangle=0,\, \forall g(x)\in\D(\mathbb{R}).
\ena
First of all we notice that if $g(x)\in\D(\mathbb{R})$ then $A^\dag g(x)\in\D(\mathbb{R})$ also.

Due to the definition of $\tau(z)$, \eqref{46ba}, and to (\ref{add2}),
we have 
\beano \langle g,A \tau(z)\rangle=N_{\varphi}\sqrt{2\pi^{1/2}}e^{\nu(z)}\left[\sum_{k\geq0}\cosh(r)\frac{\sqrt{2k}\sqrt{(2k)!}}{k!}\eta(z)^k\langle g_+,e_{2k-1}\rangle\right.+\\
\left.\sum_{k\geq0}e^{i\theta}\sinh(r)\frac{\sqrt{2k+1}\sqrt{(2k)!}}{k!}\eta(z)^k\langle g_+,e_{2k+1}\rangle\right]
\enano
or simply 
\bea
\langle  g, A \tau(z)\rangle=
\sum_{k\geq0} \alpha_\tau(2k)\langle g_+,e_{2k-1}\rangle+\beta_\tau(2k)\langle g_+,e_{2k+1}\rangle\label{47fg2}
\ena
where
\beano
\alpha_\tau(n)&=&N_{\varphi}e^{\nu(z)}\cosh(r) \sqrt{2\pi^{1/2}n}\,\frac{\sqrt{n!}}{{(n/2)!}}\eta(z)^{n/2},\\
\beta_\tau(n)&=&N_{\varphi}e^{\nu(z)}e^{i\theta}\sinh(r)\sqrt{2\pi^{1/2}(n+1)}\, \frac{\sqrt{n!}}{{(n/2)!}}\eta(z)^{n/2},
\enano
for even $n$, and the for $k=0$ the starting element of the series \eqref{47fg2} is simply $\beta_\tau(0)\langle g_+,e_{1}\rangle$. 
Hence the above series contains only elements proportional to the various $\langle g_+,e_{2n+1}\rangle$. Each of these terms comes from the element of index $n$ via $\beta_\tau(2n)\langle g_+,e_{2n+1}\rangle$ and the element of index $n+1$ via 
$\alpha_\tau(2(n+1))\langle g_+,e_{2(n+1)-1}\rangle$.
Easy investigation shows that 
\beano
\alpha_\tau(2n+2)+\beta_\tau(2n)=0,\qquad\forall n\geq 0,
\enano 
so that all the terms in the series sum up to zero as we wanted to prove.
In a similar way one can show that $\langle B^\dagger g,\kappa_z \rangle=0,\, \forall g\in\D(\mathbb{R})$.

\end{proof}

Of course this proposition, which is in agreement with what we have discussed in the first part of this section, and the fact that $\tau(z)$ and $\kappa(z)$ are always well defined for all $z\in\mathbb{C}$, suggests to look directly the solution of the differential equations $A\tau(z)=0$ and $B^\dagger\kappa(z)=0$, with $A$ and $B$ given in \eqref{add1}-\eqref{add2}. Standard computations lead to
\bea\label{bss}
\tau(z)=N_{\tau}\text{exp}\left(-d_\tau(z)\frac{{\beta_{a}^2(x)}}{2}\right),\quad
\kappa(z)=N_{\kappa}\text{exp}\left(-d_\kappa(z)\frac{\overline{\beta_{a}^2(x)}}{2}\right)\overline{\beta_{a}^\prime(x)},
\ena
where $d_\tau(z)=1/(1-e^{i\theta}\tanh(r)),\,d_\kappa(z)=-1+d_\tau(z)$, and $N_{\tau},N_{\kappa}$ are normalization factors.
As $\Re\{d_{\tau}(z)\}>0$ for all $z\in\mathbb{C}$, in the case $\beta_a(x)\in\mathbb{R}$  we see that $\tau(z)\in\Lc^2(\mathbb{R})$,  whereas for $(\cos(\theta)-\tanh(r))\tanh(r)<0$ we have $\Re\{d_\kappa(z)\}<0$, so that, depending also on the behaviour of $\beta_a^\prime(x)=1/\alpha(x)$, we could have $\kappa(z)\notin\Lc^2(\mathbb{R})$. This, once again, emphasizes the necessity of working in a distributional way. It is important to stress also that the closed forms presented in equations \eqref{bss} are undoubtedly easier to handle than the series expansions given in equations \eqref{46ba}-\eqref{46bb}, which we consider here as the natural counterpart of similar series expansions for standard squeezed states appearing in the context of a single mode of the electromagnetic field.

\section{Conclusions}\label{sectconcl}
In this paper we have used the general theory of weak pseudo-bosons in the context of ECSusy. We have shown that $\D$-PBs provide examples of ladder operators obeying an extended version of $\mathfrak{su}(1,1)$. In particular we have shown that ours is a {\em doubled extended} settings: the first extension is based on the lost of self-adjointness in going from $\mathfrak{su}(1,1)$  to the extended $\mathfrak{su}(1,1)$  considered in Section \ref{sect2}. The second extension consists in leaving $\ltwo$, while keeping {\em compatibility} between eigenstates of different number-like operators connected by an adjoint operation.

{Some preliminary results on squeezed states have also been discussed. In particular, we have proven that these states exist, in a weak sense, working in the compatibility setting we have introduced. This result is an intriguing aspect of the  weak formulation connected with the pseudo-bosonic structures, although  several key aspects still require further exploration. Primarily, we want to address the definition of \textit{coherent squeezed} states, which, in standard quantum mechanics, are typically obtained through the combined action of squeezing and displacement operators over the ground state. Furthermore, a challenging problem regards the proper definition of these operators in a distributional settings when not dealing with the conventional $\ltwo$ space. Notably, the recent work  \cite{BG22} introduced a displacement-like operator that acts on the ground states $\varphi_0$ and $\psi_0$ of the families $\F_{\varphi}$ and $\F_{\psi}$, leaving open the problem of extending its definition to encompass the entire sets. We aim in some future works to address these open problems to improve our understanding of squeezed states in the distributional frameworks. We also plan to look for possible physical appearances of our squeezed states in (possibly) gain-loss systems, or in other systems driven by non self-adjoint Hamiltonians.}

\section*{Acknowledgements}
{
	F.B. and F.G. acknowledge partial financial support from Palermo University (via FFR2023 "Bagarello" and FFR2023 "Gargano"). L.S. acknowledges financial support from \textit{Progetto REACTION “first and euRopEAn siC
		eighT Inches pilOt liNe”}. All authors acknowledge partial financial support from G.N.F.M. of the INdAM.}

\section*{Ethics approval and consent to participate}
Not applicable.
\section*{Consent for publication}
{Not applicable.}
\section*{Availability of data and materials}
{Not applicable.}
\section*{Competing interests}
{The authors declare that they have no competing interests.}
\section*{Funding}
{Not applicable.}
\section*{Authors' contributions}
All authors have contributed equally. All authors read and approved the final manuscript.

\end{document}